\documentclass[a4paper,11pt]{article}
\usepackage{latexsym,amssymb,amsfonts,amsmath,amsthm,amsmath,mathrsfs,
bbold,hyperref,amsfonts,dsfont}
\usepackage[dvips]{graphicx}
\usepackage{authblk}
\usepackage{xcolor}
\usepackage{soul}

\hypersetup{colorlinks,linkcolor=blue,citecolor=blue,urlcolor=blue}

\newtheorem{proposition}{Proposition}

\usepackage[font={small,it}]{caption}
\usepackage{soul}

\input{./Qcircuit}

\usepackage{color}

\title{Experimental Implementation of Quantum Walks on\\ IBM Quantum Computers}
\author{F. Acasiete$^{1}$, F. P. Agostini$^{2}$, J. Khatibi Moqadam$^{1}$, and R. Portugal$^{1}$\\
{\small
$^1$National Laboratory of Scientific Computing - LNCC, \\
Av. Get{\'u}lio Vargas 333, Petr{\'o}polis, RJ, 25651-075, Brazil \\
\vspace{5pt}
}
{\small
$^2$National Institute of Metrology, Quality, and Technology - Inmetro, \\
Av. Nossa Senhora das Gra\c{c}as 50, Duque de Caxias, RJ, 25250-020, Brazil \\
} 
} 

\begin{document}
\maketitle

\begin{abstract}
The development of universal quantum computers has achieved remarkable success in recent years, culminating with the quantum supremacy reported by Google. Now is possible to implement short-depth quantum circuits with dozens of qubits and to obtain results with significant fidelity. Quantum walks are good candidates to be implemented on the available quantum computers.  In this work, we implement discrete-time quantum walks with one and two interacting walkers on cycles, two-dimensional lattices, and complete graphs on IBM quantum computers. We are able to obtain meaningful results using the cycle, the two-dimensional lattice, and the complete graph with 16 nodes each, which require 4-qubit quantum circuits up to depth 100.

\


\end{abstract}

\section{Introduction}

Quantum walks are considered the quantum analogue of classical walks, which are useful for developing classical randomized algorithms~\cite{MR96}. Quantum walks have already proved useful for designing quantum algorithms~\cite{Amb07a}. The most general definition of a quantum walk on a graph demands that its time evolution obey the laws of quantum mechanics and is constrained by graph locality~\cite{Por18book}.

A quantum computer that can realize computational tasks that the largest supercomputers available nowadays cannot simulate has been recently built~\cite{GoogleQC}, opening a broad highway to implement quantum walks efficiently. The implementation of quantum walks on quantum computers requires $\log_2 N$ qubits, where $N$ is the number of nodes of the graph around which the walkers ramble. On the other hand, in many forms of implementing quantum walk in laboratories, the number of devices the experimenter puts on the table scales with the number of nodes of the graph~\cite{DW09}, though the required resources are decreased by resorting to classical realization of quantum walks with optical systems~\cite{Regensburger2011photon,Schetal12,LMNPGBJS19}.
Quantum computers provide us with an exponential advantage, however, indicating by the noisy intermediate-scale quantum (NISQ) era, the available quantum computers are prone to errors above the threshold required to implement quantum error correcting codes \cite{Preskill2018quantum}. That obviously limits the depth of the circuits that would simulate the quantum walk reasonably well.

The time evolution of quantum walks on graphs can be continuous or discrete~\cite{ADZ93,FG98}. In the discrete-time case, the very first model~\cite{ADZ93}, originally called ``random quantum walks'', nowadays known as coined quantum walks, has an internal coin space, which was considered mandatory for many years until two alternative restricted coinless models were proposed: Szegedy's~\cite{Sze04a} and Patel et.~al's~\cite{PRR05}. Szegedy's model is defined on bipartite graphs, and Patel et.~al's model on hypercubic lattices. These two models are particular cases of the staggered quantum walk model~\cite{PSFG16}, which is defined on arbitrary graphs. To define discrete-time quantum walks on arbitrary graphs without resorting to internal spaces, the number of local unitary operators, the product of which is the evolution operator, must be larger than two depending on the graph type because locality demands that the vertex set must be partitioned into cliques, inevitably leading to the notion of graph tessellation cover, and to the results proved in~\cite{ACFKMPP20}. For instance, the evolution operator of discrete-time quantum walk on two-dimensional lattices must be the product of at least four local operators~\cite{PF17}. 

In this work, we implement staggered quantum walks (SQWs)~\cite{PSFG16} on cycles, two-dimensional lattices with cyclic boundary conditions, and complete graph on IBM quantum computers. The evolution operator of a SQW on a graph is obtained using a graph tessellation cover, where a tessellation is a partition of the vertex set into cliques and a tessellation cover is a set of tessellations that covers all the edges of the graph. When implementing on IBM quantum computers, the evolution operator must be decomposed in terms of basic gates, CNOT and 1-qubit rotations. The most important gate in quantum walk implementations is the multi-controlled Toffoli gate, whose decomposition has been widely studied~\cite{HLZWW17,LL16,NC00}. In our case, we use an alternative version to this gate to shorten its decomposition, which is crucial in NISQ systems. We are able to implement one step of the quantum walk on a 16-node cycle, two steps of two interacting walkers on a 4-node cycle, and one step of the quantum walk on a 16-node two-dimensional lattice with cyclic borders. Those results improve earlier attempts using IBM quantum computers~\cite{BCS18,GZ19,Sha19} and are comparable with the size of cycles used in direct laboratory experiments~\cite{matjeschk2012experimental,flurin2017observing,dadras2019experimental}. We have also implemented quantum walk-based search algorithms on complete graphs with 8 and 16 vertices, which are more efficient that their classical versions in terms of oracle call. As far as we know, it seems that ours is the first implementation of a modified version of Grover's algorithm with four qubits on public access quantum computers with high fidelity (72.1\%) and faster than classical random search algorithms (see also the discussion at the concluding section of~\cite{SOM20}).

The structure of this paper is as follows. Sec.~\ref{sec:cycle} describes the dynamics of one walker and two interacting walkers on the $N$-cycle and their implementation on IBM quantum computers using 4 qubits. Sec.~\ref{sec:grid} describes the dynamics of one walker on a $N$-torus (cyclic two-dimensional lattice) and its implementation on IBM quantum computers using 4 qubits. Sec.~\ref{sec:search} presents the implementation of quantum walk-based search algorithms on complete graphs. Sec.~\ref{sec:conc} describes our conclusions.



\section{Quantum walk on the cycle}\label{sec:cycle}




Consider a $N$-cycle whose vertices are labeled by $0, \ldots, N-1$ and assume that $N$ is even.  A tessellation cover $\{\mathcal{T}_\alpha,\mathcal{T}_\beta\}$ for this graph is depicted in Fig.~\ref{fig:non-uniform-tiles-N=8}, where $\mathcal{T}_\alpha=\{\alpha_x:0\le x\le N/2-1\}$, $\mathcal{T}_\beta=\{\beta_x:0\le x\le N/2-1\}$, $\alpha_x=\{2x,2x+1\}$, and $\beta_x=\{2x+1,2x+2\}$.  The arithmetic is performed modulo $N$. 
Each vertex $v$ is associated with a canonical basis vector $\ket{v}$ in a Hilbert space $\mathscr{H}^{N}$, whose computational basis is $\{\ket{x}:x=0,\ldots,N-1 \}$. Each tile $\alpha_x$ ($\beta_x$) of tessellation $\mathcal{T}_\alpha$ ($\mathcal{T}_\beta$) is associated with a unit vector $\ket{\alpha_x}$ ($\ket{\beta_x}$) in $\mathscr{H}^{N}$ as follows
\begin{align}
\left|\alpha_{x}\right\rangle  &= \frac{\ket{2x}+\ket{2x+1}}{\sqrt{2}}, \label{eq:sm_alpha_latt} \\
\left|\beta_{x}\right\rangle   &= \frac{\ket{2x+1}+\ket{2x+2}}{\sqrt{2}}. \label{eq:sm_beta_latt}
\end{align}
Using these vectors, we define projectors $\sum_x\ket{\alpha_x}\bra{\alpha_x}$ and  $\sum_x\ket{\beta_x}\bra{\beta_x}$, which allow us to define the following Hermitian and unitary operators:
\begin{align}
H_0 &= 2\sum_{x=0}^{{N}/{2}-1} \ket{\alpha_x}\bra{\alpha_x}-\mathds{I}, \label{eq:sm_H0} \\
H_1 &= 2\sum_{x=0}^{{N}/{2}-1} \ket{\beta_x}\bra{\beta_x}-\mathds{I}, \label{eq:sm_H1}
\end{align}
where $\mathds{I}$ is the identity operator in $\mathscr{H}^N$, whose dimension should be clear from the context.
The evolution operator for the SQW on the cycle is given by
\begin{equation}
	{{U}}=\textrm{e}^{{-i}\theta H_1}\textrm{e}^{{-i}\theta H_0},\label{eq:unit}
\end{equation}
where $\theta$ is an angle~\cite{POM17}. The quantum walk dynamics is generated by repeatedly applying the evolution operator, at discrete time steps, on an initial state.

\begin{figure}[h!]
\centering
\begingroup%
  \makeatletter%
  \providecommand\color[2][]{%
    \errmessage{(Inkscape) Color is used for the text in Inkscape, but the package 'color.sty' is not loaded}%
    \renewcommand\color[2][]{}%
  }%
  \providecommand\transparent[1]{%
    \errmessage{(Inkscape) Transparency is used (non-zero) for the text in Inkscape, but the package 'transparent.sty' is not loaded}%
    \renewcommand\transparent[1]{}%
  }%
  \providecommand\rotatebox[2]{#2}%
  \newcommand*\fsize{\dimexpr\f@size pt\relax}%
  \newcommand*\lineheight[1]{\fontsize{\fsize}{#1\fsize}\selectfont}%
  \ifx\svgwidth\undefined%
    \setlength{\unitlength}{140.74441144bp}%
    \ifx\svgscale\undefined%
      \relax%
    \else%
      \setlength{\unitlength}{\unitlength * \real{\svgscale}}%
    \fi%
  \else%
    \setlength{\unitlength}{\svgwidth}%
  \fi%
  \global\let\svgwidth\undefined%
  \global\let\svgscale\undefined%
  \makeatother%
  \begin{picture}(1,0.96747633)%
    \lineheight{1}%
    \setlength\tabcolsep{0pt}%
    \put(0,0){\includegraphics[width=\unitlength,page=1]{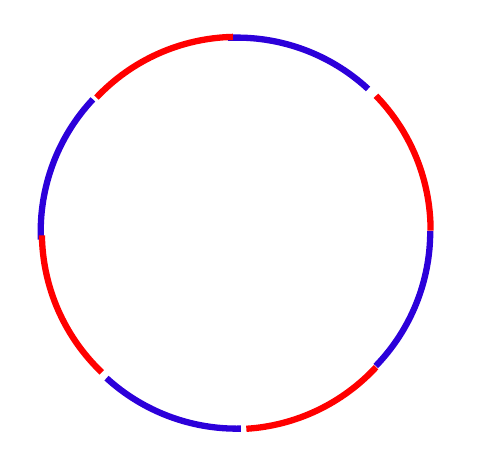}}%
    \put(0.631233,0.89421971){\color[rgb]{0,0,0}\makebox(0,0)[lt]{\lineheight{1.25}\smash{\begin{tabular}[t]{l}$\ket{+}$\end{tabular}}}}%
    \put(0.87961914,0.32378456){\color[rgb]{0,0,0}\makebox(0,0)[lt]{\lineheight{1.25}\smash{\begin{tabular}[t]{l}$\ket{+}$\end{tabular}}}}%
    \put(0.2526378,0.04292616){\color[rgb]{0,0,0}\makebox(0,0)[lt]{\lineheight{1.25}\smash{\begin{tabular}[t]{l}$\ket{+}$\end{tabular}}}}%
    \put(-0.01180225,0.66087856){\color[rgb]{0,0,0}\makebox(0,0)[lt]{\lineheight{1.25}\smash{\begin{tabular}[t]{l}$\ket{+}$\end{tabular}}}}%
    \put(0.43631638,0.925017){\color[rgb]{0,0,0}\makebox(0,0)[lt]{\lineheight{1.25}\smash{\begin{tabular}[t]{l}$0$\end{tabular}}}}%
    \put(0.7701803,0.80462604){\color[rgb]{0,0,0}\makebox(0,0)[lt]{\lineheight{1.25}\smash{\begin{tabular}[t]{l}$1$\end{tabular}}}}%
    \put(0.910835,0.48451839){\color[rgb]{0,0,0}\makebox(0,0)[lt]{\lineheight{1.25}\smash{\begin{tabular}[t]{l}$2$\end{tabular}}}}%
    \put(0.7921422,0.15022061){\color[rgb]{0,0,0}\makebox(0,0)[lt]{\lineheight{1.25}\smash{\begin{tabular}[t]{l}$3$\end{tabular}}}}%
    \put(0.46143672,0.0056781){\color[rgb]{0,0,0}\makebox(0,0)[lt]{\lineheight{1.25}\smash{\begin{tabular}[t]{l}$4$\end{tabular}}}}%
    \put(0.14498181,0.1509273){\color[rgb]{0,0,0}\makebox(0,0)[lt]{\lineheight{1.25}\smash{\begin{tabular}[t]{l}$5$\end{tabular}}}}%
    \put(0.01351133,0.46679581){\color[rgb]{0,0,0}\makebox(0,0)[lt]{\lineheight{1.25}\smash{\begin{tabular}[t]{l}$6$\end{tabular}}}}%
    \put(0.14780754,0.78407739){\color[rgb]{0,0,0}\makebox(0,0)[lt]{\lineheight{1.25}\smash{\begin{tabular}[t]{l}$7$\end{tabular}}}}%
    \put(0,0){\includegraphics[width=\unitlength,page=2]{tessellation1D.pdf}}%
    \put(0.86906543,0.66536013){\color[rgb]{0,0,0}\makebox(0,0)[lt]{\lineheight{1.25}\smash{\begin{tabular}[t]{l}$\ket{+}$\end{tabular}}}}%
    \put(0.649684,0.05917712){\color[rgb]{0,0,0}\makebox(0,0)[lt]{\lineheight{1.25}\smash{\begin{tabular}[t]{l}$\ket{+}$\end{tabular}}}}%
    \put(0.00830654,0.28088969){\color[rgb]{0,0,0}\makebox(0,0)[lt]{\lineheight{1.25}\smash{\begin{tabular}[t]{l}$\ket{+}$\end{tabular}}}}%
    \put(0.22146999,0.89639859){\color[rgb]{0,0,0}\makebox(0,0)[lt]{\lineheight{1.25}\smash{\begin{tabular}[t]{l}$\ket{+}$\end{tabular}}}}%
    \put(0,0){\includegraphics[width=\unitlength,page=3]{tessellation1D.pdf}}%
  \end{picture}%
\endgroup%
\caption{Tessellation cover of the 8-cycle showing the vectors associated with each tile $\{v,w\}$, where $\ket{+}=(\ket{v}+\ket{w})/\sqrt 2$.}
\label{fig:non-uniform-tiles-N=8}
\end{figure}

In matrix form, operators $H_0$ and $H_1$ are given by
\begin{align}
H_0 &= \mathds{I} \otimes X,\label{eq:sm_H0_matrix} \\
H_1 &= \begin{bmatrix} 
                    0          &           \bf{}                       &     1      \\ 
                 \bf{}         &     \mathds{I} \otimes X      &   \bf{}    \\
                    1          &           \bf{}                       &     0
\end{bmatrix},\label{eq:sm_H1_matrix}
\end{align}
where $X=\left(\begin{smallmatrix}0&1 \\1 & 0 \end{smallmatrix} \right)$ and the empty entries are 0. Using that
\begin{equation}\label{eq:R_x}
R_x(\theta)=\exp(-i\theta X/2) = \begin {bmatrix} \cos \frac{\theta}{2} &-i\sin \frac{\theta}{2} \\ 
\noalign{\medskip}-i\sin \frac{\theta}{2} &\cos \frac{\theta}{2} 
\end {bmatrix},
\end{equation}
we obtain the evolution generated by $H_0$ and $H_1$, respectively, as
\begin{equation}\label{eq:U_0}
{{U}}_0 = \mathds{I} \otimes R_x(2\theta),
\end{equation}
which is a block-diagonal matrix, and
\begin{equation}\label{eq:U_1}
{{U}}_1 = \begin{bmatrix} 
\cos\theta   &                                  & -i\sin\theta       \\ 
     &      \mathds{I} \otimes R_x(2\theta)   & \           \\
-i\sin\theta &                                &  \cos\theta
\end{bmatrix},
\end{equation}
which is a permutation of the rows and columns of ${{U}}_0$.
In fact, the permutation (and circulant) matrix
\begin{equation}\label{eq:permutation}
P \,=\,  
\sum_x \ket{x+1}\bra{x}
\,=\,
\begin{bmatrix} 
0  &      &          &         & 1  \\
1  &   0  &          &         &    \\ 
   &   1  &          &         &    \\
   &      &   \ddots &  \ddots &    \\
0  &      &          & 1       & 0
\end{bmatrix},
\end{equation}
transforms ${{U}}_0$ to ${{U}}_1$ via the similarity transformation ${{U}}_1=P^{-1} {{U}}_0 P$. 
The SQW evolution operator can, therefore, be written as ${{U}}=P^{-1}{{U}}_0P\,{{U}}_0$. $P$ shifts the walker to the right and $P^{-1}$ shifts to the left. This dynamics is similar to the split-step protocol introduced by Kitagawa \textit{et. al}~\cite{KRBD10}, whose implementation using photonic technology was decribed in~\cite{KBFRBKADW12}. In the split-step protocol, the shift to the right (left) occurs only if the particle spin is up (down) otherwise the particle stays put. In the staggered dynamics, the spin plays no role and the shift to the right or left is unconditional.

In the staggered model, the unit vectors that are associated with the tiles can be different from the one described by Eqs.~(\ref{eq:sm_alpha_latt}) and~(\ref{eq:sm_beta_latt}). In this case, the new local evolution operators ${{U}}_0$ and ${{U}}_1$ have the same structure described in Eqs.~(\ref{eq:U_0}) and~(\ref{eq:U_1}), but they use new $2\times 2$ matrices in place of $R_x(2\theta)$. For instance, if the unit vector associated with the first tile of tessellation ${\mathcal T}_\beta$ is $(\ket{1}\pm i\ket{2})/\sqrt 2$, the corresponding block $R_x(2\theta)$ in the expression of ${{U}}_1$ is replaced by $\pm R_y(2\theta)=\pm\left(\begin{smallmatrix}\cos\theta&-\sin\theta \\ \sin\theta & \cos\theta \end{smallmatrix} \right)$. We use those tiles to shorten the decomposition of ${{U}}_1$ in terms of basic gates.

\subsubsection*{Two interacting quantum walkers}

Let us address the dynamics of a 2-particle quantum walk on a cycle with a special type of interaction between the walkers. The evolution operator of two independent quantum walks on a cycle is the tensor product 
$${{U}}_\text{free}=\left({{U}}^{(1)}_1 {{U}}^{(1)}_0\right) \otimes \left({{U}}^{(2)}_1 {{U}}^{(2)}_0\right)$$
of two 1-particle quantum walks. The resulting operator belongs to the Hilbert space $\mathscr{H}^{N} \otimes\mathscr{H}^{N}$.

Now, suppose the walkers interact when they are simultaneously at the same vertex of the cycle, and consider the interaction described by a phase shift $\phi$ on top of the free evolution operator. The modified evolution operator is
\begin{equation}
{{U}} = {{U}}_\text{free} \, R,
\end{equation}
where
\begin{equation}
\label{eq:interaction_operator}
R \,\ket{x_1}\ket{x_2} = \begin{cases} 
      e^{i\phi}\ket{x_1}\ket{x_2}, &\text{if}\;\;\; x_1=x_2,\\
      \ket{x_1}\ket{x_2}, &\text{otherwise.} 
                         \end{cases}
\end{equation} 
$R$ is a diagonal matrix, whose diagonal entries are either 1 or $e^{i\phi}$.

An alternative interacting model, similar to the one used when designing quantum search algorithms on graphs, is
\begin{equation}
\label{eq:interaction_operator_2}
R \,\ket{x_1}\ket{x_2} = \begin{cases} 
      e^{i\phi}\ket{x_1}\ket{x_2}, &\text{if}\;\;\; x_1=x_2=x^0,\\
      \ket{x_1}\ket{x_2}, &\text{otherwise,} 
                         \end{cases}
\end{equation} 
where $x^0$ is a marked vertex. The decomposition of operator $R$ in the alternative model in terms of basic gates is shorter than the original one.

\subsection{Decomposition of the evolution operator}\label{sec:decomp}

In this section, we present the decomposition of the SQW evolution operators assuming that $N=2^n$ for an integer $n$. The Hilbert space $\mathscr{H}^N$ is spanned by the computational basis of $n$ qubits. 
Note that each vertex of the cycle is represented by a computational basis vector $\ket{q_0\ldots q_{n-1}}$, where $q_i$ are qubits.

\subsubsection*{Decomposition of the permutation matrix $P$}

The matrix representation of the operator ${{U}}_0$, given by Eq.~(\ref{eq:U_0}), has the decomposition ${{U}}_0=\mathds{I}_2^{\otimes n-1} \otimes R_x(2\theta)$, that is, $(n-1)$ $2\times 2$-identity operators acting on the first $n-1$ qubits and a rotation $R_x(2\theta)$ on the last qubit (see the central part of Fig.~\ref{fig:U_1}). Fig.~\ref{fig:U_1} also depicts the circuit that implements ${{U}}_1$. The circuit of $P$ is shown at the left-hand part and its inverse $P^{-1}$ at the right-hand part.

\begin{figure}[h!]
$$
\Qcircuit @C=0.9em @R=0.8em { 
\lstick{\scriptstyle q_0}& \targ    &  \qw      & \qw       & \qw                   & \qw                             & \qw                   &  \qw      & \qw       & \targ     & \qw \\
\lstick{\scriptstyle q_1}& \ctrl{-1}& \targ     & \qw       & \qw                   & \qw                             & \qw                   &  \qw      & \targ     & \ctrl{-1} & \qw \\ 
\lstick{\scriptstyle q_2}& \ctrl{-1}& \ctrl{-1} & \targ     & \qw                   & \qw                             & \qw                   & \targ     & \ctrl{-1} & \ctrl{-1} & \qw \\
\lstick{\scriptstyle q_3}& \ctrl{-1}& \ctrl{-1} & \ctrl{-1} & \gate{\scriptstyle X} &\gate{\scriptstyle R_x(2\theta)} & \gate{\scriptstyle X} & \ctrl{-1} & \ctrl{-1} & \ctrl{-1} & \qw  \\
& &\hspace{0.8cm}\scriptstyle P \gategroup{1}{2}{4}{5}{0.8em}{--}& & &\scriptstyle {{U}}_0 \gategroup{1}{6}{4}{6}{0.4em}{--}& &\hspace{1cm}\scriptstyle P^{-1} \hspace{0.3cm} \gategroup{1}{7}{4}{10}{0.8em}{--}& & &
}  
$$
\caption{Circuit of the operator ${{U}}_1$, including $P$, ${{U}}_0$, and $P^{-1}$, for a quantum walk on a cycle with 16 vertices (4 qubits).}
\label{fig:U_1}
\end{figure}
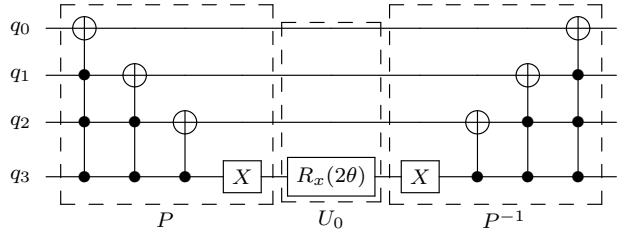

The decomposition of matrix $P$ uses multi-controlled Toffoli gates~\cite{DW09}, which are defined in the following way. Suppose that $C_{i_1,i_2,\ldots}(X_j)$ represents a multi-controlled Toffoli gate with control qubits ${q_{i_1}},{q_{i_2}},\ldots$ and the target qubit ${q_j}$. The action of $C_{i_1,i_2,\ldots}(X_j)$ is nontrivial only on $\ket{q_j}$, which is given by
\begin{equation}
\label{eq:generalized_Toff}
C_{i_1,i_2,\ldots}(X_j)\ket{q_{i_1},q_{i_2}...}\ket{q_j}=\ket{q_{i_1},q_{i_2}...}X^{q_{i_1}\cdot q_{i_2} \cdots}\ket{q_j}=\ket{q_{i_1},q_{i_2}...}\ket{q_j\oplus (q_{i_1}\cdot q_{i_2}\cdots)},
\end{equation}
that is, the state of qubit ${q_j}$ changes only if $q_{i_1}$, $q_{i_2},\ldots$ are all set to 1. Note that, in Fig.~\ref{fig:U_1}, the first (top) qubit of the circuit is ${q_0}$ and the last (bottom) one  is ${q_{n-1}}$. The correctness proof of this decomposition is shown in Appendix~\ref{appen:proof}.

\subsubsection*{Decomposition of the multi-controlled Toffoli gate}

To decompose the multi-controlled Toffoli gate $C_{0,...,n-2}(X_{n-1})$, which has $n-1$ control qubits $q_0$, ..., $q_{n-2}$ and one target qubit $q_{n-1}$, we initially use the identity 
$$C_{0,...,n-2}(X_{n-1})=H_{n-1} C_{0,...,n-2}(Z_{n-1})H_{n-1},$$ 
where $H_{n-1}$ is the Hadamard gate acting on qubit $q_{n-1}$, and then we focus on the method to decompose $C_{0,...,n-2}(Z_{n-1})$. Fig.~\ref{fig:economicdecomp} shows how to decompose  $C_{0,...,n-2}(Z_{n-1})$ in terms of a sequence of multi-controlled $R_z(\theta)$,
\begin{equation}\label{eq:R_z}
R_z(\theta)=\exp(-i\theta Z/2) = 
\begin {bmatrix} \text{e}^{-i\theta/2} &0 \\ 
\noalign{\medskip}0 & \text{e}^{i\theta/2}
\end {bmatrix},
\end{equation}
where $\theta=\pi/2^{n-j}$, $j=1,...,n$.

\begin{figure}[h!]
\[
\Qcircuit @C=0.9em @R=0.6em { 
\lstick{\scriptstyle q_0}      & \ctrl{1}               & \qw &        && &\gate{\scriptstyle R_z\big(\frac{\pi}{2^{n-1}}\big)}&\ctrl{1} &\ctrl{1} &\qw&          &&&\ctrl{1}  &\qw \\
\lstick{\scriptstyle q_1}      & \ctrl{1}               & \qw &        && &\qw      &\gate{\scriptstyle R_z\big(\frac{\pi}{2^{n-2}}\big)}&\ctrl{1} &\qw&          &&&\ctrl{1}  &\qw \\ 
\lstick{\scriptstyle q_2}      & \ctrl{1}               & \qw & \equiv && &\qw      &\qw      &\gate{\scriptstyle R_z\big(\frac{\pi}{2^{n-3}}\big)}&\qw&          &&&\ctrl{1}  &\qw \\
                               &                        &     &        && &         &         &                                                    &   &          &&&          &    \\
                               & \vdots                 &     &        && &         &\vdots   &                                                    &   & {\ddots} &&& \vdots   &    \\ 
                               &                        &     &        && &         &         &                                                    &   &          &&&          &    \\
                               & \ctrl{1}               & \qw &        && &         &         &                                                    &   &          &&& \ctrl{1} &\qw \\
\lstick{\scriptstyle q_{n-1}}  & \gate{\scriptstyle Z}  & \qw & && &\qw &\qw &\qw &\qw& &&&\gate{\scriptstyle R_z\big(\frac{\pi}{2^0}\big)}&\qw \gategroup{1}{14}{8}{14}{0.8em}{--} \\
}
\]
\caption{Decomposition of $C_{0,...,n-2}(Z_{n-1})$.}
\label{fig:economicdecomp}
\end{figure}
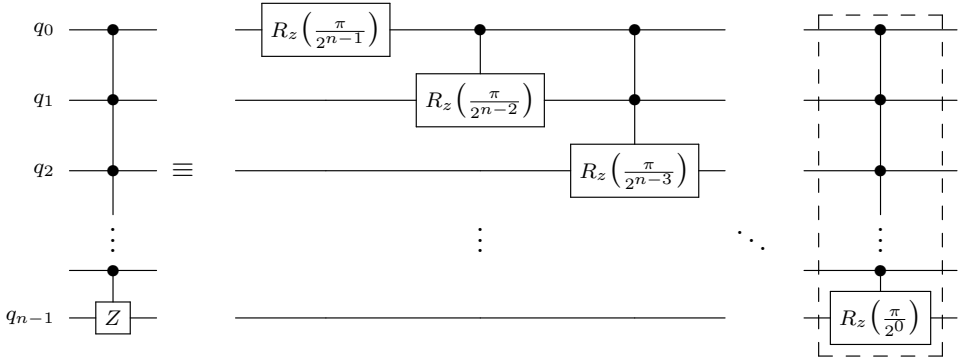

An example of the decomposition of the multi-controlled $R_z(\pi)$ gate for $n=4$ is depicted in Fig.~\ref{fig:mcRz}. This is the last multi-controlled gate in the decomposition of the CCCZ gate. The generic decomposition of the multi-controlled $R_z(\pi/2^j)$ gate is given in terms of an alternated sequence of CNOT and $u_1(\pm\pi/2^{j-n})$ gates as described by function \verb|new_mcrz| in Appendix~\ref{appen_1}, where
\begin{equation}\label{eq:gate_U_1}
u_1(\theta)=
\begin {bmatrix} 1 &0 \\ 
\noalign{\medskip}0 & \text{e}^{i\theta}
\end {bmatrix}.
\end{equation}
Note that $R_z(\theta)$ and $u_1(\theta)$ differ by a global phase and sometimes can be interchanged.
The positions of the CNOT controls (except for the first CNOT) in the decomposition of the multi-controlled $R_z(\pi/2^j)$ gate are given by function 
\begin{equation}\label{a(k)}
a(k)\,=\,\log_2[k-k\&(k-1)],
\end{equation}
where \& is the bitwise AND operator. For instance, the positions of the CNOT controls in Fig.~\ref{fig:mcRz} starting from the second is 0,1,0,2,0,1,0, which correspond to $a(1)$, ..., $a(7)$. The decomposition used in this work is useful only when the number of qubits is small, since the decomposition size increases as an exponential function in terms of the number of qubits. When the number of qubits is large, it is recommended to use ancilla qubits~\cite{HLZWW17}.

\begin{figure}[h!]
\begin{equation*}
    \Qcircuit @C=0.5em @R=0.9em {
	 	& \ctrl{1}&\qw&             &&      && & \qw & \qw & \ctrl{3} & \qw & \qw & \qw & \ctrl{3} & \qw & \qw & \qw & \ctrl{3} & \qw & \qw & \qw & \ctrl{3} & \qw & \qw & \qw\\
	 	& \ctrl{1}&\qw&             &&      && & \qw & \qw & \qw & \qw & \ctrl{2} & \qw & \qw & \qw & \qw & \qw & \qw & \qw & \ctrl{2} & \qw & \qw & \qw & \qw & \qw\\
	 	& \ctrl{1}&\qw&             &&\equiv&& & \ctrl{1} & \qw & \qw & \qw & \qw & \qw & \qw & \qw & \ctrl{1} & \qw & \qw & \qw & \qw & \qw & \qw & \qw & \qw & \qw\\
	 	& \gate{\scriptstyle R_z(\pi)}&\qw&      &&      && & \targ & \gate{\scriptstyle -\frac{\pi}{8}} & \targ & \gate{\scriptstyle +\frac{\pi}{8}} & \targ & \gate{\scriptstyle -\frac{\pi}{8}} & \targ & \gate{\scriptstyle +\frac{\pi}{8}} & \targ & \gate{\scriptstyle -\frac{\pi}{8}} & \targ & \gate{\scriptstyle +\frac{\pi}{8}} & \targ & \gate{\scriptstyle -\frac{\pi}{8}} & \targ & \gate{\scriptstyle +\frac{\pi}{8}} & \qw & \qw\\
	 }
\end{equation*}
\caption{Decomposition of multi-controlled $R_z(\pi)$, where \fbox{$\pm\frac{\pi}{8}$} is $u_1(\pm\pi/8)$. }
\label{fig:mcRz}
\end{figure}
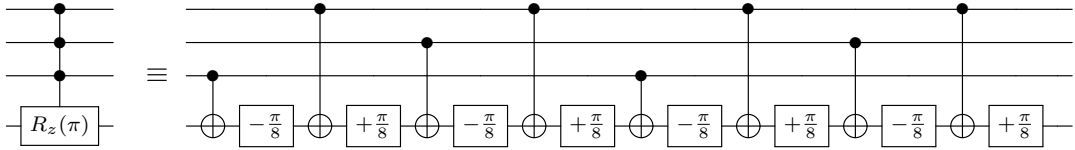

\subsubsection*{Alternative version to matrix $P$}

In this subsection we describe an alternative version to matrix $P$, whose decomposition in terms of basic gates is shorter. The strategy is to use a sequence of vectors $\ket{\pm}=(\ket{v}\pm \ket{w})/\sqrt 2$ and  $\ket{\pm i}=(\ket{v}\pm i\ket{w})/\sqrt 2$ as the unit vectors associated with the tiles of tessellation $\mathcal{T}_\beta$, where $v$ and $w$ are the vertices of the tile. The sequence is described by function $a(k)$ modulo 4 starting with $k=1$, where $a(k)$ is given by (\ref{a(k)}), and by Table~\ref{table:ak}, which associates each value $a(k) \mod 4$ with a unit vector in the set $\{\ket{\pm},\ket{\pm i}\}$.
\begin{table}[h!]
\begin{center}
\begin{tabular}{|c|l|c|l|}
\hline
 $a(k)\mod 4$ & vector & Hamiltonian &  evolution \\
\hline
 0 & $\ket{+}$ &    X     &$R_x(2\theta)$ \\
\hline
 1 & $\ket{-i}$ &     -Y   & $R_y(-2\theta)$ \\
\hline
 2 & $\ket{-}$ &       -X     &$R_x(-2\theta)$ \\
\hline
 3 & $\ket{+i}$ &    Y      &$R_y(2\theta)$ \\
\hline
\end{tabular}
\end{center}
\caption{Association between the set of tiles of tessellation $\mathcal{T}_\beta$ and the unit vectors $\ket{\pm}=(\ket{v}\pm \ket{w})/\sqrt 2$,  $\ket{\pm i}=(\ket{v}\pm i\ket{w})/\sqrt 2$. The third and the forth columns describe the sub-matrices of $H_1$ and ${{U}}_1$, respectively.}\label{table:ak}
\end{table}

In order to obtain the new local operator ${{U}}_1$, which uses the new unit vectors, we replace all multi-controlled Toffoli gates (2 or more controls) by multi-controlled $C(R_x(\pi))$ gates. Fig.~\ref{fig:new_P} describes the circuit of the new version for 4 qubits. Note that the multi-controlled $C(R_x(\pi))$ gates can be expressed as $HC(R_z(\theta))H$, where $H$ is the Hadamard gate. The decomposition of the new version in terms of basic gates can be accomplished by using the technique shown in Fig.~\ref{fig:mcRz}. The number of CNOT gates in our decomposition of the alternative version to $P$ is 13 for $n=4$, 
which is less than the original $P$, that has 21 CNOTs. 
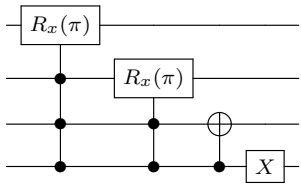
\begin{figure}[h!]
$$
\Qcircuit @C=0.5em @R=0.5em {
&                                         &                       &     &     &     \\
& \gate{\scriptstyle R_x(\pi)}            & \qw                   & \qw & \qw & \qw \\
& \ctrl{-1}& \gate{\scriptstyle R_x(\pi)} & \qw                   & \qw & \qw &     \\ 
& \ctrl{-1}& \ctrl{-1} & \targ            & \qw                   & \qw &     &     \\
& \ctrl{-1}& \ctrl{-1} & \ctrl{-1}        & \gate{\scriptstyle X} & \qw &     &     \\
}  
$$ 
\caption{Circuit of the alternative version to matrices $P$. }
\label{fig:new_P}
\end{figure}

In the one-dimensional case, there is a straightforward equivalence between the SQW and coined models. The alternative version to $P$ in ${{U}}_1$ represents a nonhomogeneous coin, that is, a different coin for each vertex. ${{U}}_0$ on the other hand represents a lazy (when $\theta<\pi/2$) flip-flop shift operator.

\subsection{Implementations on IBM quantum computers}

We use Qiskit\footnote{\url{https://qiskit.org/}} to build the circuits of the evolution operator and to run the experiments that are shown in this section. The experiments must be run when the error rate of the quantum computers are as low as possible, otherwise they output useless results.

\subsubsection*{Results for one quantum walker}

Fig.~\ref{fig:4q1pqx2} depicts the probability distribution after one step of a staggered quantum walk with the modified tiles that are associated with the alternative version to matrix $P$. The walker's initial position is the origin. The action of the evolution operator spreads the position among vertices 0, 1, 2, and 15. 

\begin{figure}[h!]
\centering
\includegraphics[scale=0.7]{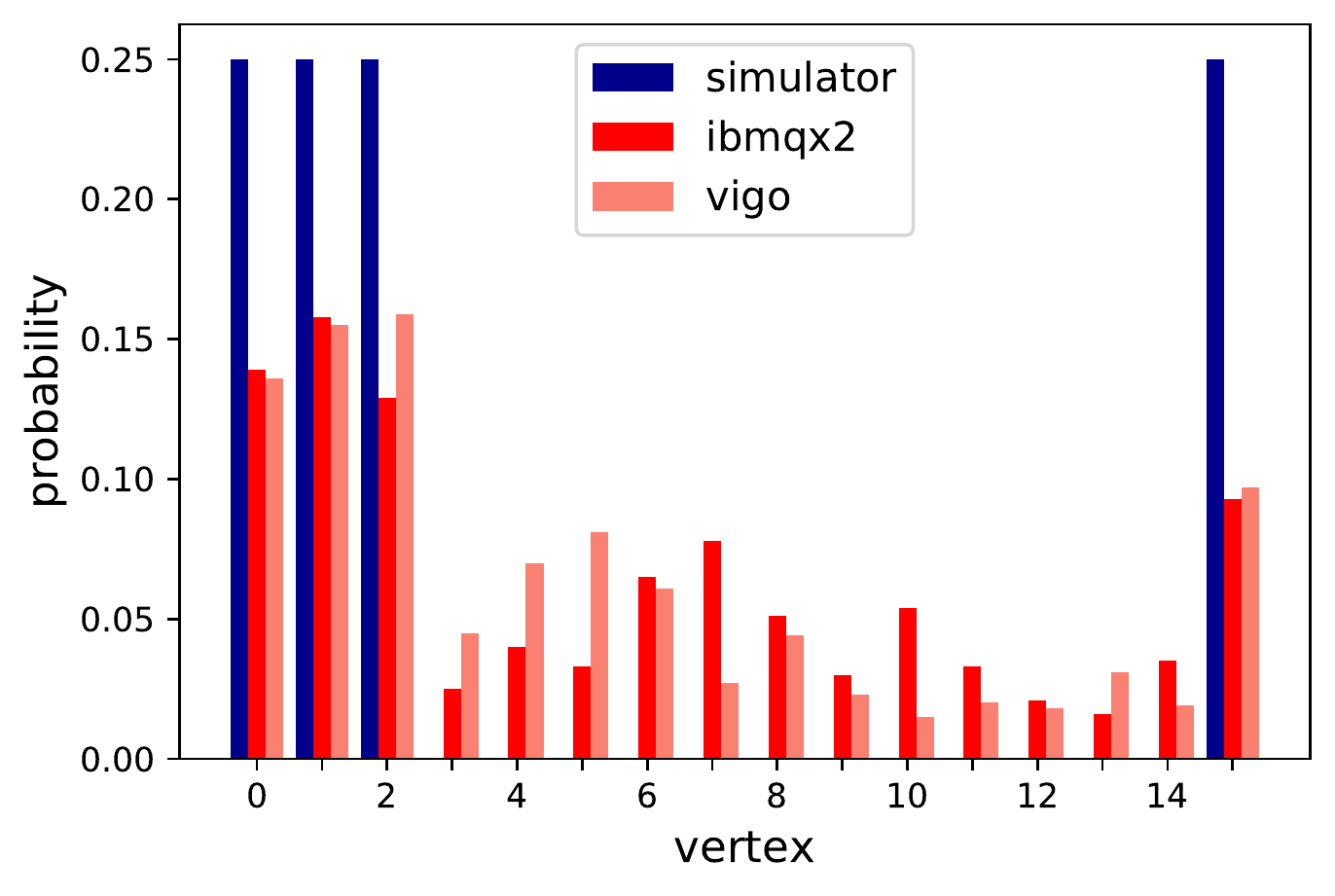}
\caption{Probability distribution after one step using simulation (blue), ibmqx2 (red), and vigo (salmon) quantum computers employing 4 qubits. }
\label{fig:4q1pqx2}
\end{figure}

The blue bars represent the simulated probability distribution and the red bars represent the result of the experiment in ibmqx2 (red) and ibmq\_vigo (salmon) quantum computers. The fidelities between the simulated and actual results are given in Table~\ref{table:fid1}.

\begin{table}[h!]
\begin{center}
\begin{tabular}{|c|c|c|}
\hline
 fidelity & ibmqx2 & vigo \\
\hline
 $1-d$ & 0.519  & 0.547 \\
\hline
 $1-h$ & 0.468  & 0.486 \\
\hline
\end{tabular}
\end{center}
\caption{Fidelities between the probability distributions generated by the quantum computer $p$ and the exact simulation $q$ for one walker on a 16-vertex cycle, where the total variation distance $d$ and the Hellinger distance $h$ are given by $d=\frac{1}{2}\sum_x |p_x-q_x|$ and $h^2=\frac{1}{2} \sum_x \left(\sqrt{p_x}-\sqrt{q_x}\right)^2$.}\label{table:fid1}
\end{table}

\subsubsection*{Results for two interacting quantum walkers}

Fig.~\ref{fig:interacting2} depicts the probability distribution of two interacting quantum walkers up to two steps on a 4-node cycle. We have used the interaction given by~(\ref{eq:interaction_operator_2}) taking the node with label 3 as marked. The initial state is $(x_1,x_2)=(0,2)$, which is obtained with high fidelity. After the first step, the positions of both walkers spread along the whole cycle, and then they interact at node with label 3. After the second step, we have an entangled state that shows that the walkers are at $(x_1,x_2)=(1,3)$ with high probability (77.6\%).

\begin{figure}[h!]
\centering
\includegraphics[scale=0.37]{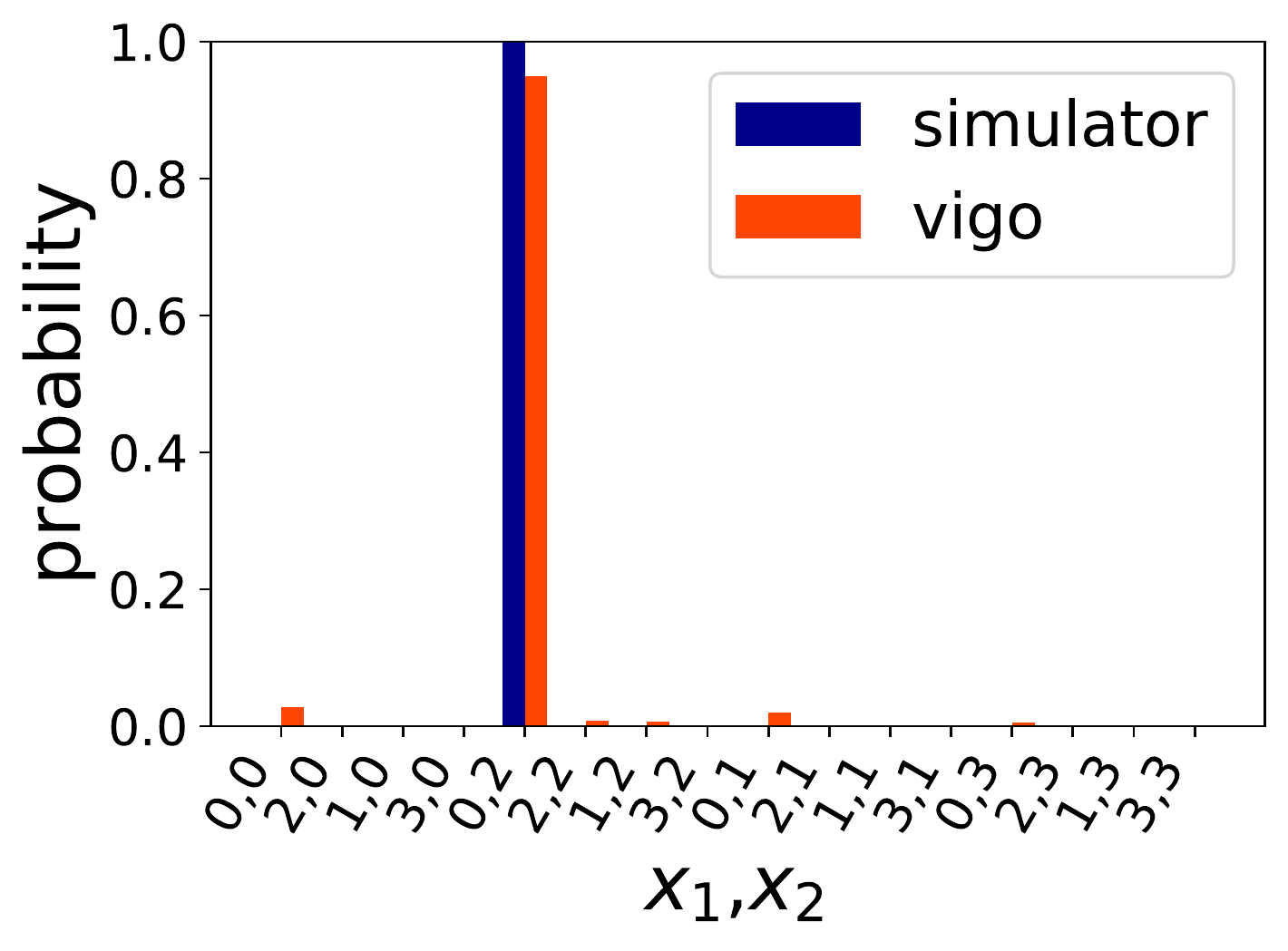}
\includegraphics[scale=0.37]{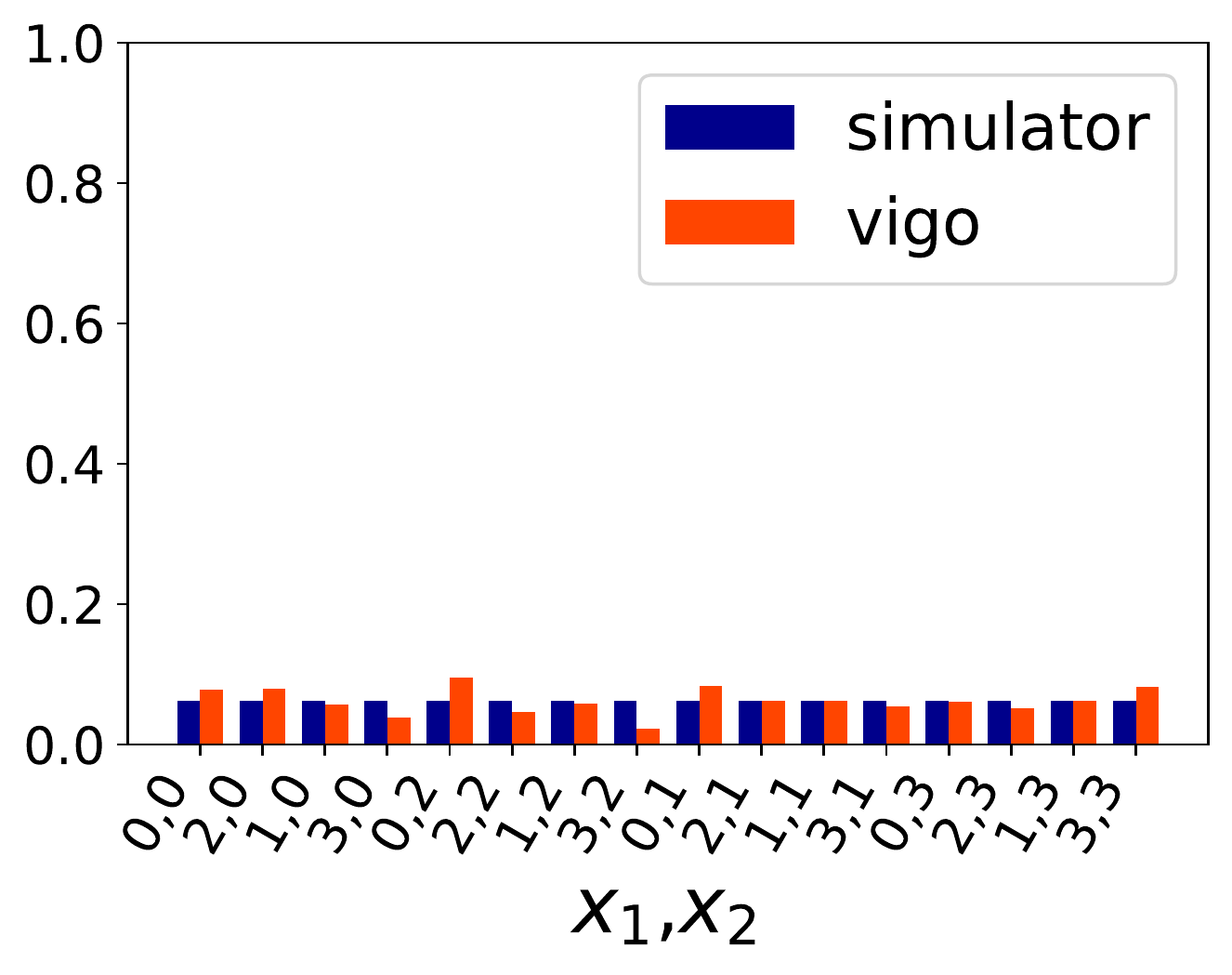}
\includegraphics[scale=0.37]{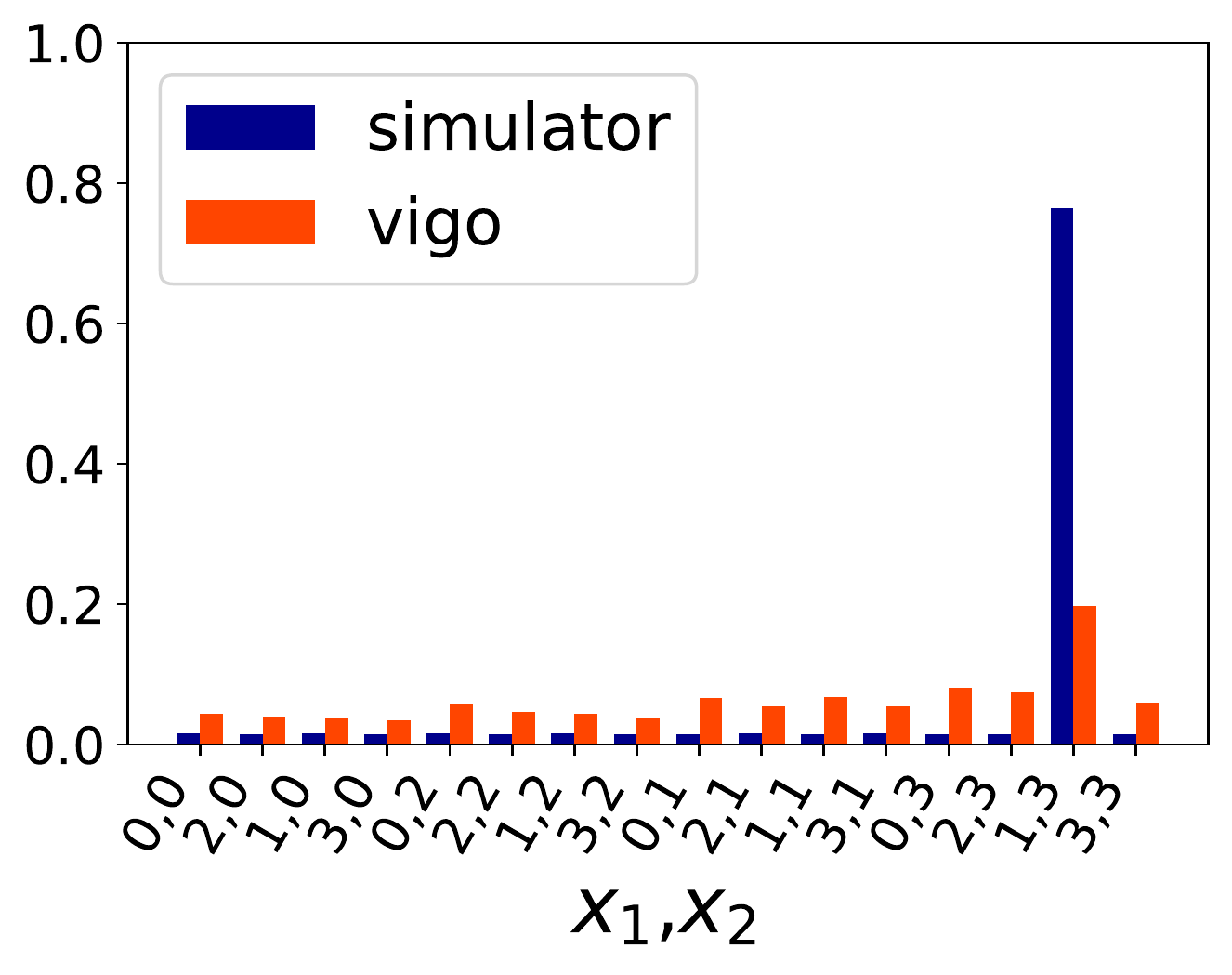}
\caption{Probability distribution of the initial state (i.s.), first and second steps using simulation (blue) and vigo quantum computer (red) for two interacting quantum walkers on a 4-vertex cycle using 4 qubits. }
\label{fig:interacting2}
\end{figure}

The fidelities between the exact calculations (blue) and the results generated by the quantum computer (red) are given in Table~\ref{table:fid2}. Although the fidelity of the second step is not high, the position of the largest peak of probability distribution obtained from the quantum computer coincides with the correct position.

\begin{table}[h!]
\begin{center}
\begin{tabular}{|c|c|c|c|}
\hline
 fidelity & i.s. & step 1 & step 2  \\
\hline
 $1-d$ & 0.950 & 0.885 & 0.433 \\
\hline
 $1-h$ & 0.814 & 0.890 & 0.572 \\
\hline
\end{tabular}
\end{center}
\caption{Fidelities between the probability distributions generated by the quantum computer and by the exact calculations.}\label{table:fid2}
\end{table}

\section{Quantum walk on the torus}\label{sec:grid}

Consider a two-dimensional square lattice with cyclic boundary conditions and ${\sqrt{N}}\times {\sqrt{N}}$ vertices labeled by $0, \ldots, {N}-1$, where $N$ is a square number. At least four tessellations are required to define the evolution operator of a  SQW on the lattice~\cite{ACFKMPP20}. There are infinite ways to tessellated the lattice. The simplest way without being trivial is depicted in Fig.~\ref{fig:lattice}.

\begin{figure}[h!]
\centering
\includegraphics[scale=0.3]{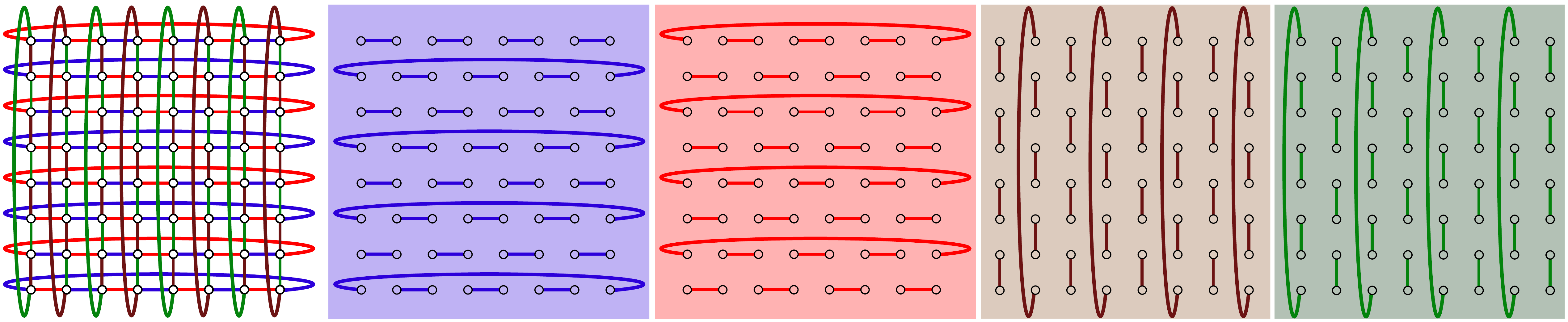}
\caption{A tessellation cover of the two-dimensional lattice with cyclic borders in the form of a 64-vertex torus. Each tessellation is associated with a local unitary operator.}
\label{fig:lattice}
\end{figure}

Using the one-dimensional SQW evolution operators (\ref{eq:U_0}) and (\ref{eq:U_1}), and labeling the vertices from top-left to down-right row-by-row, we obtain the matrix form for the two-dimensional SQW operators
\begin{align}
{{U}}_{00} &= \left(\mathds{I} \otimes \begin{bmatrix}1  &  0\\0  &  0\\ \end{bmatrix}\right) \otimes {{U}}_0 + \left(\mathds{I} \otimes \begin{bmatrix}0  &  0\\0  &  1\\ \end{bmatrix}\right) \otimes {{U}}_1 \label{eq:U_0_hor},\\
{{U}}_{10} &= \left(\mathds{I} \otimes \begin{bmatrix}1  &  0\\0  &  0\\ \end{bmatrix}\right) \otimes {{U}}_1 + \left(\mathds{I} \otimes \begin{bmatrix}0  &  0\\0  &  1\\ \end{bmatrix} \right) \otimes {{U}}_0 \label{eq:U_1_hor},
\end{align}
corresponding to the blue and red (first and second) tessellations in Fig.~\ref{fig:lattice}, and
\begin{align}
{{U}}_{01} &= {{U}}_0 \otimes \left(\mathds{I} \otimes \begin{bmatrix}1  &  0\\0  &  0\\ \end{bmatrix}\right) + {{U}}_1 \otimes \left(\mathds{I} \otimes \begin{bmatrix}0  &  0\\0  &  1\\ \end{bmatrix}\right) \label{eq:U_0_ver},\\
{{U}}_{11} &= {{U}}_1 \otimes \left(\mathds{I} \otimes \begin{bmatrix}1  &  0\\0  &  0\\ \end{bmatrix}\right) + {{U}}_0 \otimes \left(\mathds{I} \otimes \begin{bmatrix}0  &  0\\0  &  1\\ \end{bmatrix}\right) \label{eq:U_1_ver},
\end{align}
for the brown and green (third and forth) tessellations in Fig.~\ref{fig:lattice}.
The evolution operator for the SQW with Hamiltonians on the lattice is given by \cite{POM17,moqadam2018boundary}
\begin{equation}
{{U}}^{\mathrm{2D}} = {{U}}_{11}{{U}}_{10}{{U}}_{01}{{U}}_{00}.
\label{eq:2DSQW}
\end{equation}
The same evolution operator (with $\theta=\pi/4$) was used in \cite{PF17} to describe a quantum walk-based search algorithm and it is related with the alternate two-step model proposed in \cite{DMB11}.

\subsection{Decomposition of the 2D SQW evolution}
Assume that $\sqrt{N}$ is a power of two. The matrix representation of the operators given by Eqs.~(\ref{eq:U_0_hor}) and (\ref{eq:U_1_hor}) has the decomposition
\begin{align}
{{U}}_{00} &= \mathds{I}\otimes \ket{0}\bra{0} \otimes {{U}}_0 + \mathds{I}\otimes \ket{1}\bra{1} \otimes P^{-1}{{U}}_0 P \nonumber\\
&=  Q_x^{-1} ( \mathds{I}\otimes {{U}}_0 ) {Q_x}, \label{eq:U_0_hor_decomp} \\
{{U}}_{10} &= \mathds{I} \otimes \ket{1}\bra{1} \otimes {{U}}_0 + \mathds{I} \otimes \ket{0}\bra{0} \otimes P^{-1}{{U}}_0 P \nonumber \\
&= (\mathds{I}\otimes  X \otimes \mathds{I}) {{U}}_{00} (\mathds{I}\otimes X \otimes \mathds{I}) \label{eq:U_1_hor_decomp},
\end{align}
where $\ket{0}\bra{0}=\left(\begin{smallmatrix}1&0\\0&0\end{smallmatrix}\right)$, $\ket{1}\bra{1}=\left(\begin{smallmatrix}0&0\\0&1\end{smallmatrix}\right)$,  $\{\ket{0},\ket{1}\}$ is the computational basis for the Hilbert space corresponding to a single qubit, and ${Q_x}$ is given by
\begin{align}
{Q_x} &= \mathds{I}\otimes\ket{0}\bra{0} \otimes \mathds{I} + \mathds{I}\otimes\ket{1}\bra{1} \otimes P. \label{eq:Q_0}
\end{align}
Operator ${Q_x}$ is a controlled-$P$ gate with the control qubit $\ket{q_{n/2-1}}$ and target qubits $\ket{q_{n/2}\ldots q_{n-1}}$. The inverse of ${Q_x}$ is obtained by replacing $P$ with $P^{-1}$ in Eq.~(\ref{eq:Q_0}).

Similarly, we find
\begin{align}
{{U}}_{01} &= {{U}}_0 \otimes \mathds{I} \otimes \ket{0}\bra{0} + P^{-1}{{U}}_0 P \otimes \mathds{I} \otimes \ket{1}\bra{1} \nonumber\\
&= Q_y^{-1} \Bigl( {{U}}_0\otimes \mathds{I} \Bigr) Q_y \label{eq:U_0_ver_decomp},\\
{{U}}_{11} &= P^{-1}{{U}}_0 P \otimes \mathds{I} \otimes \ket{0}\bra{0} + {{U}}_0 \otimes \mathds{I} \otimes \ket{1}\bra{1} \nonumber\\
&= (\mathds{I} \otimes X) {{U}}_{01} (\mathds{I} \otimes X) \label{eq:U_1_ver_decomp},
\end{align}
corresponding to operators (\ref{eq:U_0_ver}) and (\ref{eq:U_1_ver}), where
\begin{align}
Q_y &= \mathds{I} \otimes \mathds{I} \otimes \ket{0}\bra{0} + P \otimes \mathds{I} \otimes \ket{1}\bra{1}, \label{eq:R_0}
\end{align}
which is a controlled-$P$ gate with the control qubit $\ket{q_{n-1}}$ and target qubits $\ket{q_{0}\ldots q_{n/2-1}}$.
Fig.~\ref{fig:2dsqw} shows the circuit that implements the SQW evolution operator given by Eq.~(\ref{eq:2DSQW}) on a 16-vertex two-dimensional lattice.

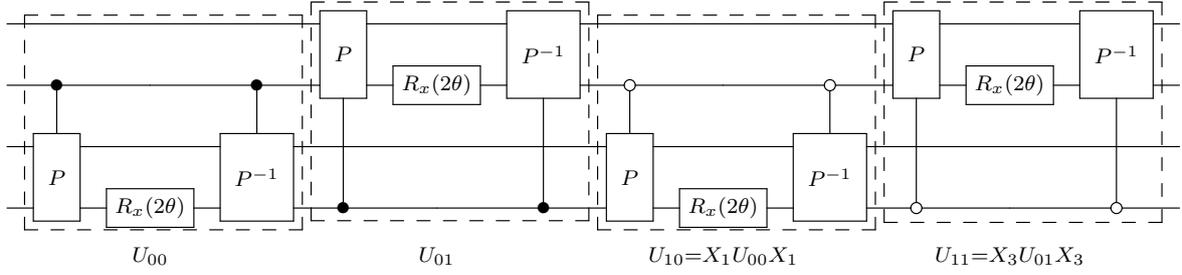
\begin{figure}[h!]
$$
\Qcircuit @C=0.9em @R=1em { 
& \qw                           & \qw                              & \qw                                & \multigate{1}{\scriptstyle P} & \qw                              & \multigate{1}{\scriptstyle P^{-1}} & \qw                           & \qw                              & \qw                                & \multigate{1}{\scriptstyle P} & \qw                              & \multigate{1}{\scriptstyle P^{-1}} & \qw  \\ 
& \ctrl{1}                      & \qw                              & \ctrl{1}                           & \ghost{\scriptstyle P}        & \gate{\scriptstyle R_x(2\theta)} & \ghost{\scriptstyle P^{-1}}        & \ctrlo{1}                     & \qw                              & \ctrlo{1}                          & \ghost{\scriptstyle P}        & \gate{\scriptstyle R_x(2\theta)} & \ghost{\scriptstyle P^{-1}}        & \qw  \\ 
& \multigate{1}{\scriptstyle P} & \qw                              & \multigate{1}{\scriptstyle P^{-1}} & \qw                           & \qw                              & \qw                                & \multigate{1}{\scriptstyle P} & \qw                              & \multigate{1}{\scriptstyle P^{-1}} & \qw                           & \qw                              & \qw                                & \qw  \\
& \ghost{\scriptstyle P}        & \gate{\scriptstyle R_x(2\theta)} & \ghost{\scriptstyle P^{-1}}        & \ctrl{-2}                     & \qw                              & \ctrl{-2}                          & \ghost{\scriptstyle P}        & \gate{\scriptstyle R_x(2\theta)} & \ghost{\scriptstyle P^{-1}}        & \ctrlo{-2}                    & \qw                              & \ctrlo{-2}                         & \qw  \\
& &{\scriptstyle {{U}}_{00}} \gategroup{1}{2}{4}{4}{0.6em}{--}& & &{\scriptstyle {{U}}_{01}} \gategroup{1}{5}{4}{7}{0.6em}{--}& & &{\scriptstyle {{U}}_{10}=X_1{{U}}_{00}X_1} \gategroup{1}{8}{4}{10}{0.6em}{--}& & &{\scriptstyle {{U}}_{11}=X_3{{U}}_{01}X_3} \gategroup{1}{11}{4}{13}{0.6em}{--}& &
}  
$$
\caption{Circuit of the two-dimensional SQW evolution operator, including ${{U}}_{00},{{U}}_{01},{{U}}_{10}$, and ${{U}}_{11}$, for a quantum walk on a lattice with 16 vertices using 4 qubits.}
\label{fig:2dsqw}
\end{figure}

\subsection{Implementations on IBM quantum computers}
In the construction of the controlled-$P$ gate, we use the alternative version to $P$ and add the control qubit to all its components, as depicted in Fig.~\ref{fig:new_cP}.
As discussed earlier, the alternative version to $P$ has a shorter decomposition in terms of basic gates at the cost of changing the unit vectors associated with the tiles of the tessellations shown in Fig.~\ref{fig:lattice}.
The corresponding sequence of unit vectors introduced by the controlled gate is described by function $a(k) \mod 4$~(\ref{a(k)}), similar to the alternative version to $P$, after interchanging $\ket{\pm}\leftrightarrow\ket{\pm i}$ in Table (\ref{table:ak}).

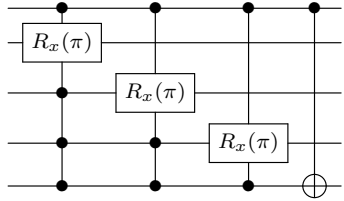
\begin{figure}[h!]
$$
\Qcircuit @C=0.5em @R=0.4em {
& \ctrl{+1}                                           & \ctrl{+2} & \ctrl{+3}& \ctrl{+4} & \qw\\
& \gate{\scriptstyle R_x(\pi)}                        & \qw       & \qw      & \qw       & \qw\\
&  \ctrl{-1}& \gate{\scriptstyle R_x(\pi)}             & \qw       & \qw      & \qw       &    \\ 
&   \ctrl{-1}& \ctrl{-1} & \gate{\scriptstyle R_x(\pi)} & \qw       & \qw      &           &    \\
&   \ctrl{-1}& \ctrl{-1} & \ctrl{-1}                    & \targ     & \qw      &           &    \\
}  
$$ 
\caption{Circuit of the alternative version to controlled-$P$ obtained from the alternative version to $P$ described in Fig.~\ref{fig:new_P}.}
\label{fig:new_cP}
\end{figure}

\begin{figure}[h!]
	\centering
	\includegraphics[scale=0.55,trim=0 0 0 45]{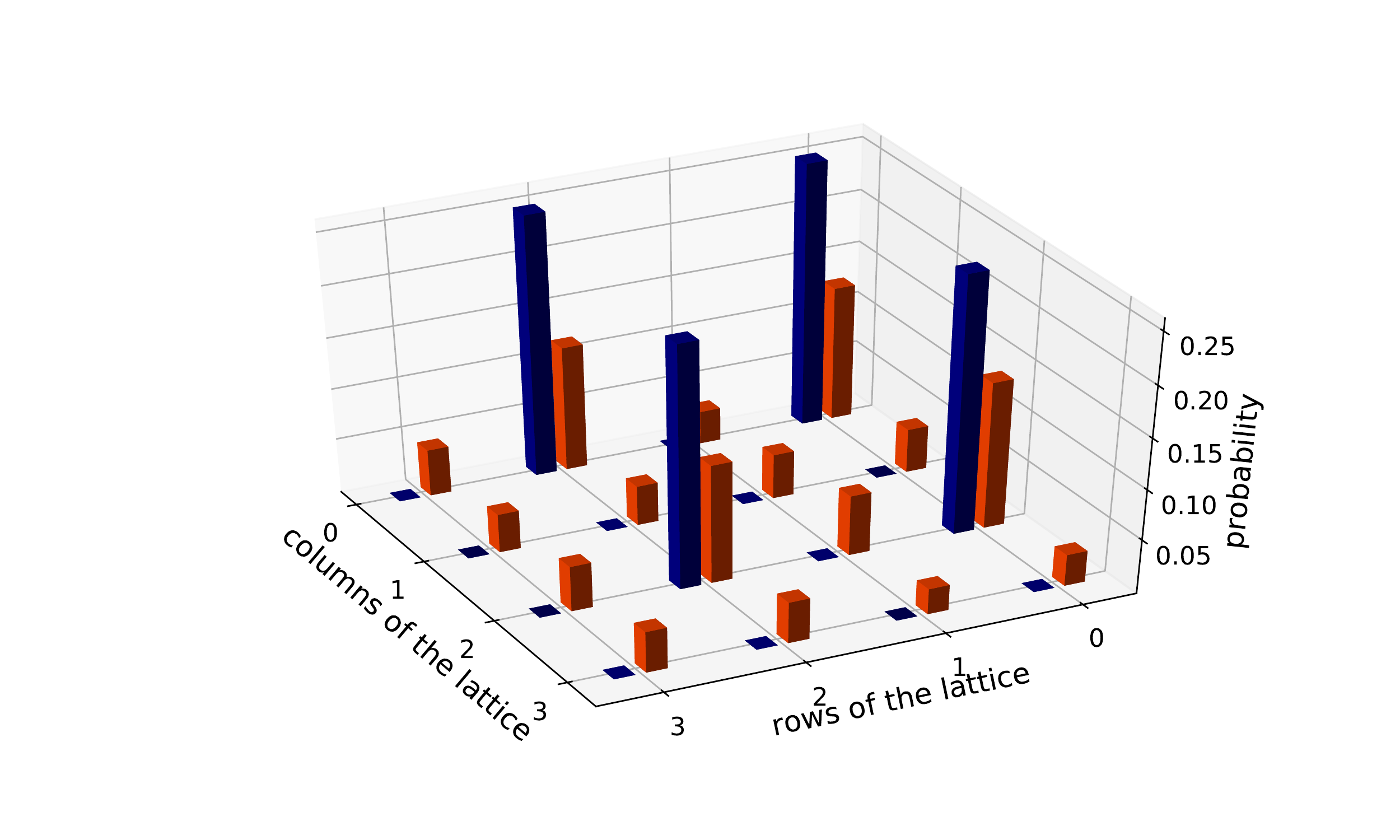}
	\caption{Probability distribution after one step of the two-dimensional SQW, with a non-local initial state, using simulation (blue) and vigo quantum computer (red) employing 4 qubits.}
	\label{fig:2dsqw-vigo}
\end{figure}

Fig.~\ref{fig:2dsqw-vigo} depicts the probability distribution after one step of the two-dimensional SQW with the modified tiles. The walker initial position is an equal superposition of all vertices with even labels with the amplitudes $1/\sqrt{8}$.
The fidelities between the simulated and actual results are given in Table~\ref{table:fid1_2D}.

\begin{table}[h!]
	\begin{center}
		\begin{tabular}{|c|c|}
			\hline
			fidelity & vigo \\
			\hline
			$1-d$ & 0.515 \\
			\hline
			$1-h$ &  0.468\\
			\hline
		\end{tabular}
	\end{center}
	\caption{Fidelities between the probability distributions generated by the vigo quantum computer and the exact simulation for one walker on a 16-vertex two-dimensional lattice.}\label{table:fid1_2D}
\end{table}

\section{Quantum walk-based spatial search}\label{sec:search}

In this section, we describe the implementation of a quantum walk-based spatial search algorithm on complete graphs with 8 ($K_8$) and 16 ($K_{16}$) vertices. Since the complete graph $K_N$ is 1-tessellable, the evolution operator of a staggered quantum walk on $K_N$ is the Grover operator $G$, given by $G=-H^{\otimes n}RH^{\otimes n}$, where $n=\log_2 N$ and
\begin{equation}
R = I-2\ket{0}\bra{0}.
\end{equation}
A quantum walk-based search algorithm uses a modified evolution operator ${{U}}'$~\cite{Por18book}, given by
\begin{equation}
{{U}}' = G\,R,
\end{equation}
when the marked vertex has label 0. The initial state $\ket{\psi_0}$ is the uniform superposition of all states of the computational basis $\ket{\psi_0}=\sum_{j=0}^{N-1}\ket{j}$, and the optimal number of steps is the closest integer to $(\pi/4)\sqrt{N}$. Note that the quantum walk-based spatial search on the complete graph is equivalent to Grover's algorithm~\cite{Gro97,Por18book}.

\subsection{Implementation on IBM quantum computers}

Since ${{U}}'= -(H^{\otimes n}R)^2$, the missing task is to find the decomposition of $R$. It is straightforward to check that
\begin{equation}
 R=X^{\otimes n}C_{0,...,n-2}(Z_{n-1})X^{\otimes n}.
 \end{equation}
The decomposition of $C_{0,...,n-2}(Z_{n-1})$ is depicted in Fig.~\ref{fig:economicdecomp}. As we have discussed in Sec.~\ref{sec:decomp}, the number of basic gates reduces if we replace $Z$ for $R_z(\pi)$ in $C_{0,...,n-2}(Z_{n-1})$. So, instead of using $R$ in our implementations, we use $R'$, which is given by
\begin{equation}
 R'=X^{\otimes n}C_{0,...,n-2}(R_z(\pi))X^{\otimes n}. 
\end{equation} 
The success probability using $R'$ is not going to be as high as is in the original algorithm, but is high enough for complete graphs up to 16 vertices.

\begin{figure}[h!]
\centering
\includegraphics[scale=0.54]{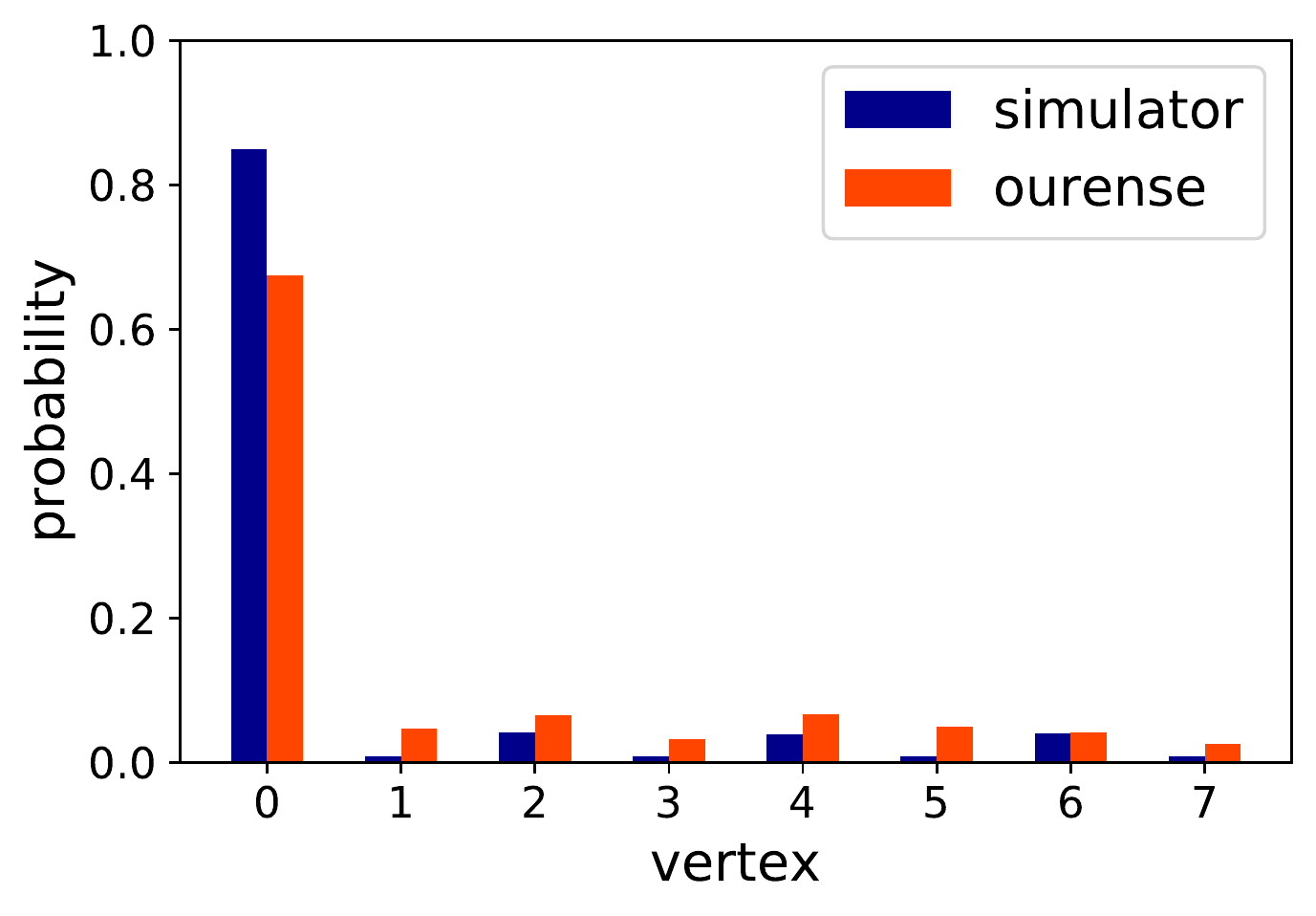}\hspace{4pt}
\includegraphics[scale=0.54]{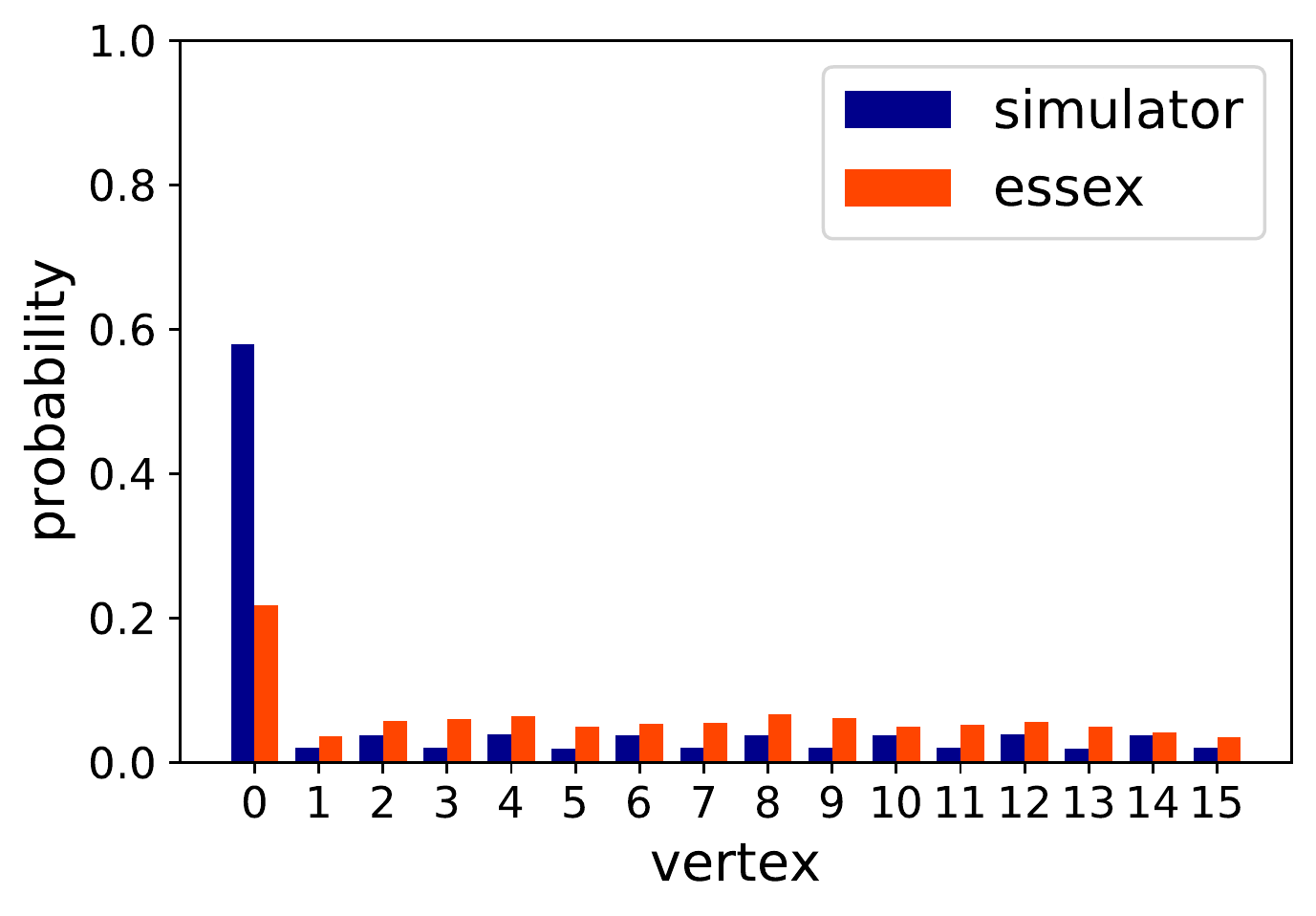}
\caption{Probability distribution of a quantum walk-based search algorithm on $K_8$ (left) and $K_{16}$ (right) after three steps using simulation (blue) and IBM quantum computers (red). }
\label{fig:spatialsearch}
\end{figure}

Fig.~\ref{fig:spatialsearch} depicts the probability distribution after three steps of the quantum walk-based search algorithm on the complete graph $K_8$ (left panel) and $K_{16}$ (right panel); Table~\ref{table:fid1_2D_search} shows the corresponding fidelities using the total variation distance $d$ and the Hellinger distance $h$. The success probability of finding the marked vertex for the $K_8$ case using a quantum computer is $0.674$ and for the $K_{16}$ case is $0.218$. These results are better than a 3-attempt random search, which has success probability $3/8=0.375$ for 8 elements and $3/16=0.187$ for 16 elements.

\begin{table}[h!]
	\begin{center}
		\begin{tabular}{|c|c|c|}
			\hline
			fidelity & $K_8$ & $K_{16}$ \\
			\hline
			$1-d$ & 0.825 & 0.639 \\
			\hline
			$1-h$ &  0.821 & 0.721 \\
			\hline
		\end{tabular}
	\end{center}
	\caption{Fidelities between the exact probability distribution (blue) and the probability distribution outputed by IBM quantum computers (red).}\label{table:fid1_2D_search}
\end{table}

\section{Conclusions}\label{sec:conc}

In this work, we have implemented the evolution operator of staggered quantum walks on cycles, two-dimensional lattices, and complete graphs on IBM quantum computers. We have shown how to decompose each local unitary operator in terms of basic gates. The multi-controlled Toffoli gate is an important building block, whose decomposition can be analyzed independently of the remaining circuit.

We have implemented the first step of a quantum walk on the 16-node cycle and 16-node two-dimensional lattice, and obtained results with fidelity around 50\%. We have implemented two steps of two interacting quantum walkers on the 4-node cycle and obtained result with fidelity around 57\%. Appendix~\ref{appen_2} describes further results for a quantum walk on the 8-node cycle, with high fidelity up to eight steps. Although high, the fidelity is not a good measure of the quality of the results for the final steps, because the exact probability distributions are almost flat, and error-dominated outputs have also almost-flat associated probability distributions.

We have implemented a quantum walk-based search algorithm on complete graphs with 8 and 16 vertices; the results are more efficient than the equivalent classical algorithms. Since the quantum walk version is equivalent to Grover's algorithm (the Qiskit program is the same), we have implemented on IBM quantum computers a modified version of Grover's algorithm that is more efficient than classical brute-force search on unsorted lists of 8 and 16 elements.

We conclude that IBM quantum computers are able to produce non-trivial results in terms of quantum walk implementations, which includes non-classical aspects of quantum mechanics, such as the entanglement between quantum walkers.

We also conclude that the staggered model is suitable for NISQ computers because the implementation of coinless models requires fewer qubits than the coined model. For instance, the implementation of two interacting quantum walkers on a 4-cycle in the coined model needs six qubits, in contrast with four qubits in the coinless case.

\section*{Acknowledgments}
The authors thank J.~Valardan, M.~A.~V.~Macedo Jr., I.~J.~Ara\'ujo Jr.,~and M.~Paredes for useful discussions. 
JKM acknowledges financial support from CNPq grant PCI-DA No.~304865/2019-2.
RP acknowledges financial support from CNPq grant No.~303406/2015-1 and Faperj grant CNE No. E-26/202.872/2018.

\appendix

\section{Appendix}\label{appen:proof}

\begin{proposition}
The decomposition of $P$ given by Eq.~(\ref{eq:permutation}) in terms of multi-controlled Toffoli gates is
\[
P = X_{k-1}\,C_{k-1}(X_{k-2})C_{k-2,k-1}(X_{k-3})\cdots C_{1,\ldots,k-1}\,(X_{0}),
\]
where $k$ is the number of qubits.
\end{proposition}
\begin{proof}
From Eq.~(\ref{eq:permutation}), we have 
\begin{equation}
P\ket{q}=\begin{cases} \ket{q+1}, & \text{if }q<N-1 \\ \ket{0}, &\text{if }q=N-1, \end{cases}
\end{equation}
where $\ket{q}$ is a generic state of the computational basis in decimal notation. The binary representation of $q$ is $(q_0\ldots q_{k-1})_2$. In the case $q=N-1$, the binary representation of the state is $\ket{(1\ldots1)_2}$ and it is straightforward to verify that the circuit of $P$ in Fig.~\ref{fig:U_1} generates the desired output state $\ket{(0\ldots 0)_2}$. Suppose that $q<N-1$. The action of $P$ on a generic qubit state is
\begin{equation}
\begin{aligned}
P\ket{q_0\cdots q_{k-1}} &= 
C_{1,\ldots,k-1}(X_{0})\ket{q_{0}}\cdots C_{k-2,k-1}(X_{k-3})\ket{q_{k-3}}\,\,C_{k-1}(X_{k-2})\ket{q_{k-2}}\,\,X_{k-1}\ket{q_{k-1}}. \nonumber\\
\end{aligned}
\end{equation}
Simplifying the right-hand side by using Eq.~(\ref{eq:generalized_Toff}) gives
\[
 \ket{q_{0}\oplus (q_1\cdots q_{k-1})}\ldots \ket{q_{k-3}\oplus(q_{k-2}\cdot q_{k-1})}\ket{q_{k-2}\oplus q_{k-1}}\ket{q_{k-1}\oplus 1}.
\]
On the other hand, the addition $q+1$ in the binary representation,  namely $(q_0 \cdots q_{k-1})_2\oplus1$ yields
\[
\begin{tabular}{cccccc}
${\color{gray}q_1\cdots q_{k-1}}$ & & ${\color{gray}q_{k-2}\cdot q_{k-1}}$ & ${\color{gray}q_{k-1}}$ & &\\
$q_{0}$ & $\ldots$ & $q_{k-3}$ & $q_{k-2}$ & $q_{k-1}$ & \\
 &  &  &  & 1&$\oplus$ \\
  \hline
$q_{0}\oplus (q_1\cdots q_{k-1})$ & $\ldots$ & $q_{k-3}\oplus(q_{k-2}\cdot q_{k-1})$ & $q_{k-2}\oplus q_{k-1}$ & $q_{k-1}\oplus 1$
\end{tabular}
\]
where the gray colored bits in the first line of the table show the carries. The result (given in the forth line of the table) is obtained by performing the addition of the rightmost bits of the table, that is, adding bits $q_{k-1}$ and 1. The result is $q_{k-1}\oplus 1$ and the carry is $q_{k-1}$, which is placed over $q_{k-2}$ as a gray colored bit. Then, bits $q_{k-1}$ and $q_{k-2}$ are added that gives $q_{k-2}\oplus q_{k-1}$ and the carry is $q_{k-2}\cdot q_{k-1}$, which is placed over $q_{k-3}$. The addition goes on until the leftmost bit is reached. The final result coincides with the action of $P$ on $\ket{q_0\ldots q_{k-1}}$, which proves the proposition.
\end{proof}

\section{Appendix}\label{appen_1}

This appendix describes function \verb|new_mcrz|, which decomposes the multi-controlled $R_z$ gate, and function \verb|new_mcz|, which decomposes the multi-controlled $Z$ gate. Those functions use the same syntax of functions \verb|mcrz| and \verb|mcz| implemented in Qiskit. Note that our implementation uses fewer CNOTs.

\begin{verbatim}
from qiskit import *
from math import pi,log
q = QuantumRegister(4)
qc = QuantumCircuit(q)

def new_mcrz(qc,theta,q_controls,q_target):
    n = len(q_controls)
    newtheta = -theta/2**n
    a = lambda n: log(n-(n&(n-1)),2)
    qc.cx(q_controls[n-1],q_target)
    qc.u1(newtheta,q_target)
    for i in range(1,2**(n)):
        qc.cx(q_controls[int(a(i))],q_target)
        qc.u1((-1)**i*newtheta,q_target)
QuantumCircuit.new_mcrz = new_mcrz

qc.new_mcrz(pi,[q[0],q[1],q[2]],q[3])
print(qc.draw())
\end{verbatim}

\begin{verbatim}
qc = QuantumCircuit(q)

def new_mcz(qc,q_controls,q_target):
    L = q_controls + [q_target]
    n = len(L)
    qc.u1(pi/2**(n-1),L[0])
    for i in range(2,n+1):
        qc.new_mcrz(pi/2**(n-i),L[0:i-1],L[i-1])
QuantumCircuit.new_mcz = new_mcz

qc.new_mcz([q[0],q[1],q[2]],q[3])
print(qc.draw())
\end{verbatim}

\section{Appendix}\label{appen_2}

Fig.~\ref{fig:8-cycle} depicts our results for a staggered quantum walk on the 8-cycle with $\theta=\pi/4$ and initial condition $(\ket{3}+\ket{4})/\sqrt 2$ using three qubits of the ourense quantum computer. The two high peaks moving in opposite direction display the well known signature of quantum walks on the one-dimensional lattice. 

\begin{figure}[h!]
\centering
\includegraphics[scale=0.33]{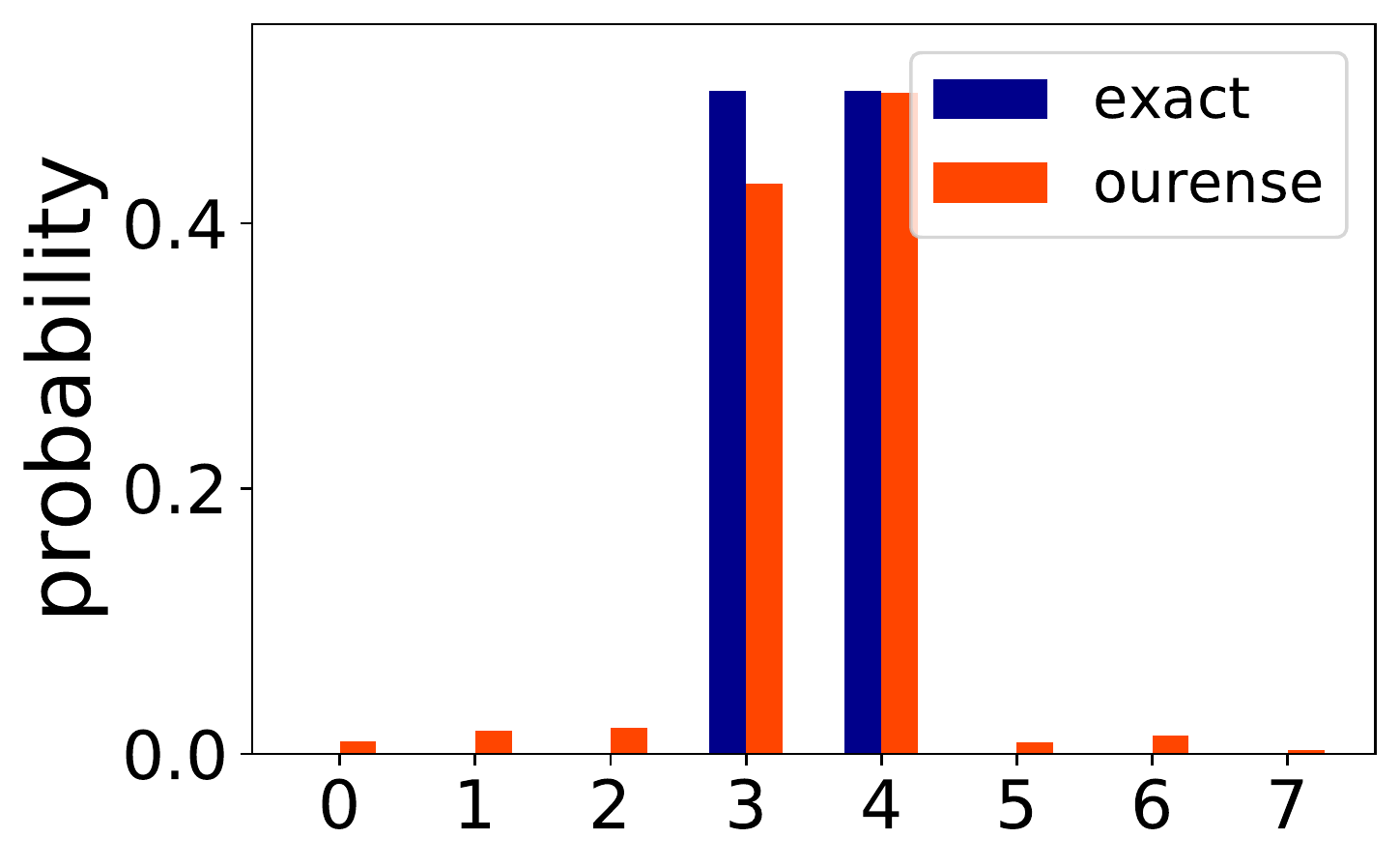}
\includegraphics[scale=0.33]{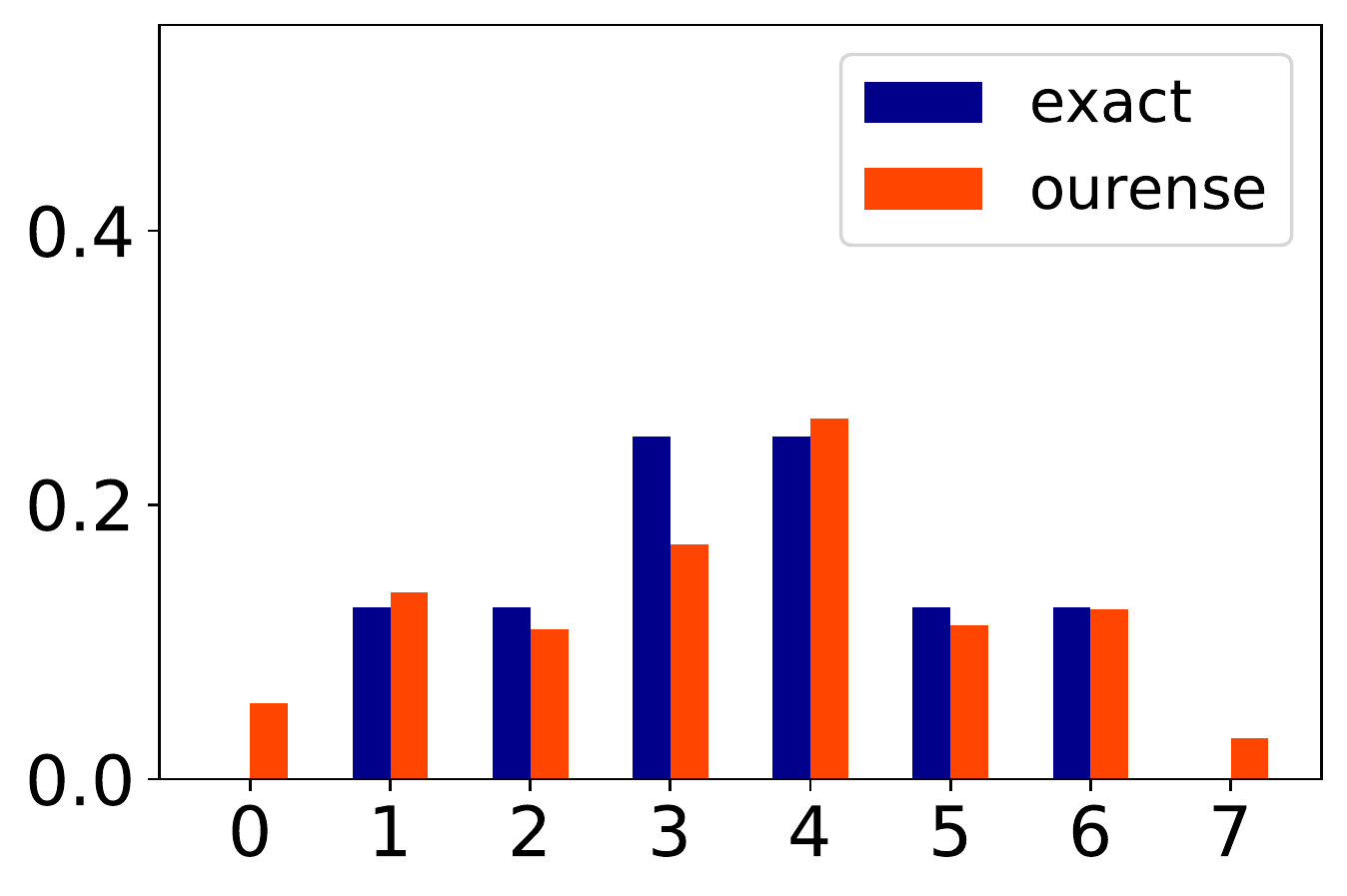}
\includegraphics[scale=0.33]{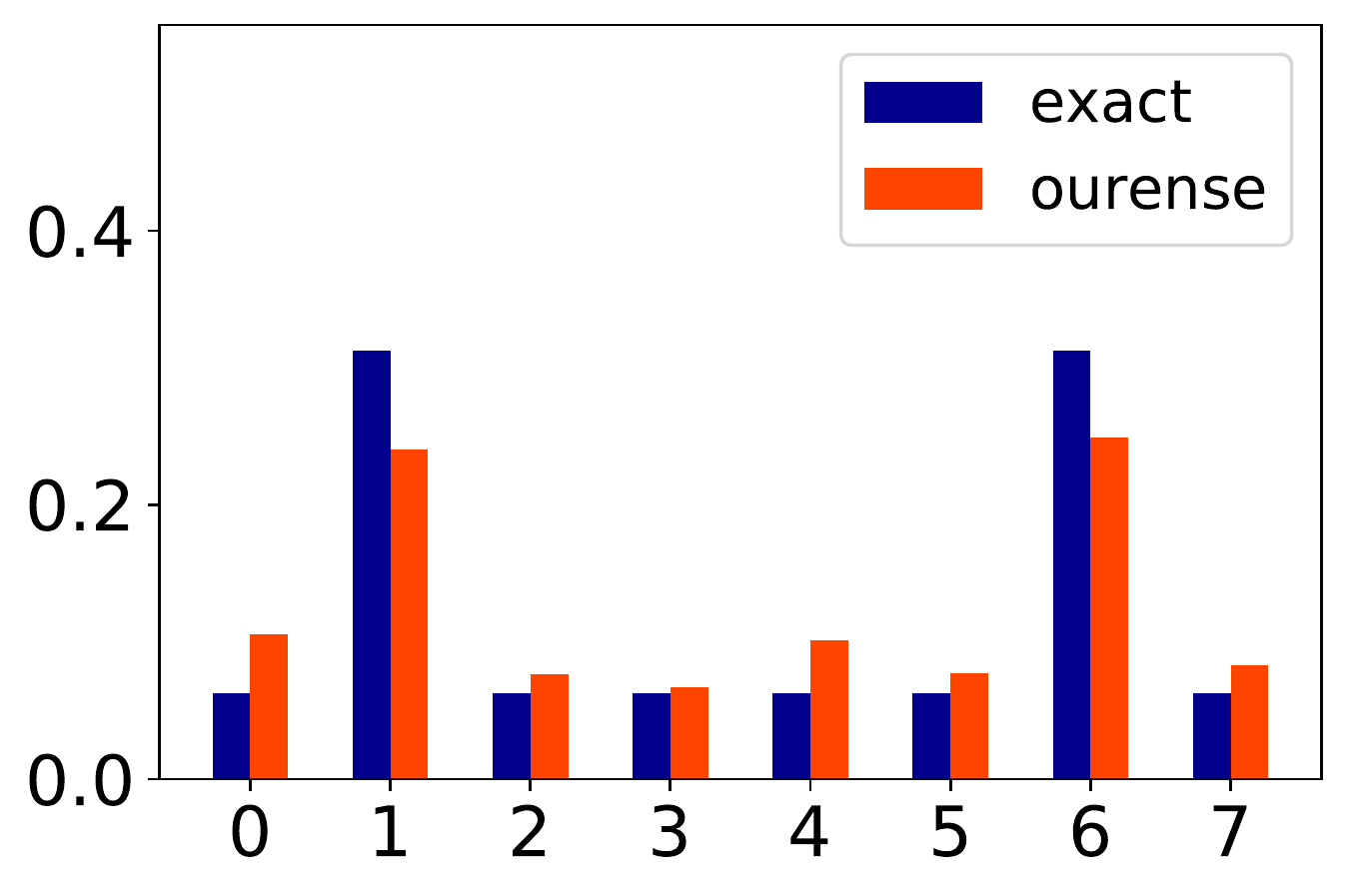}
\includegraphics[scale=0.33]{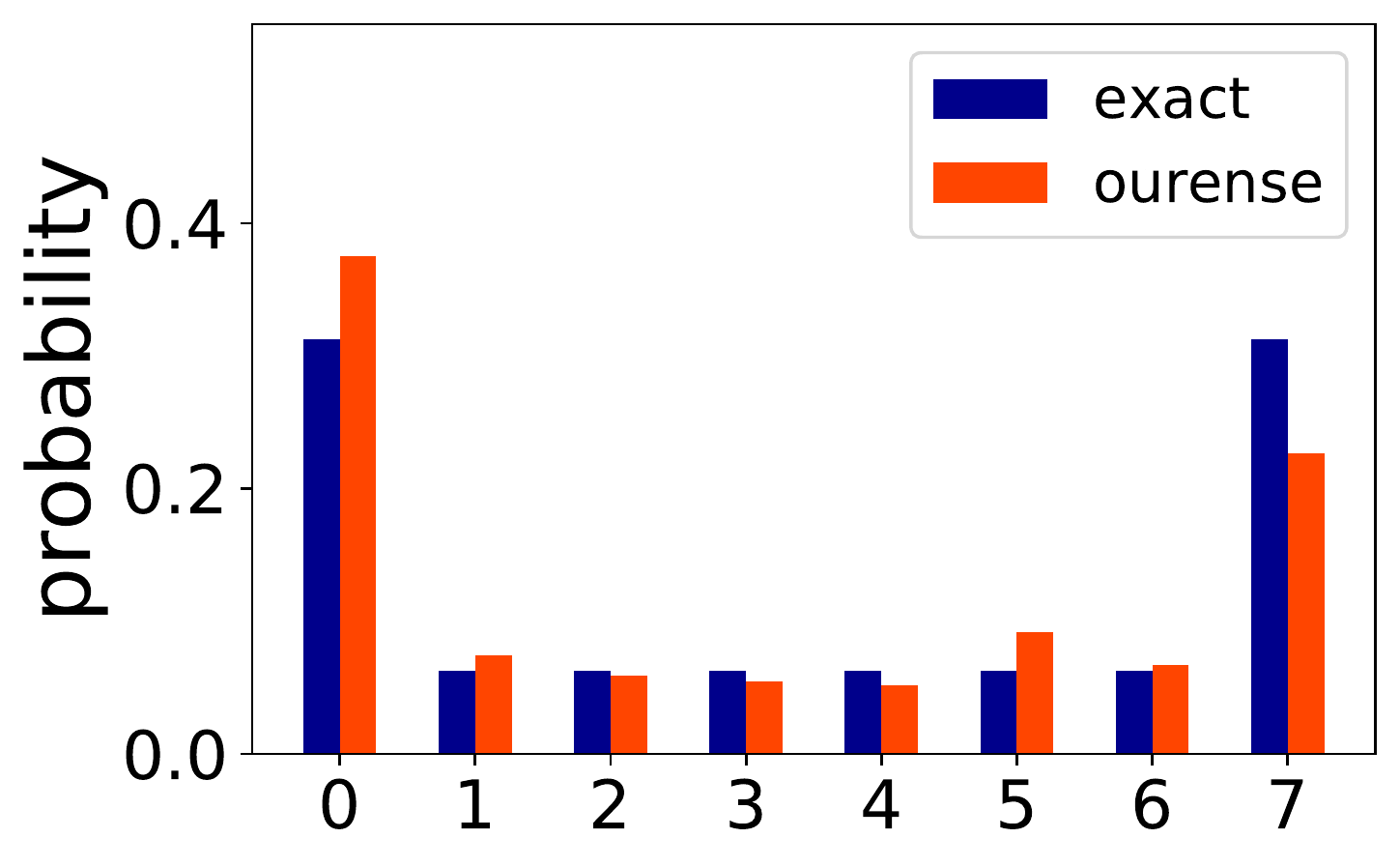}
\includegraphics[scale=0.33]{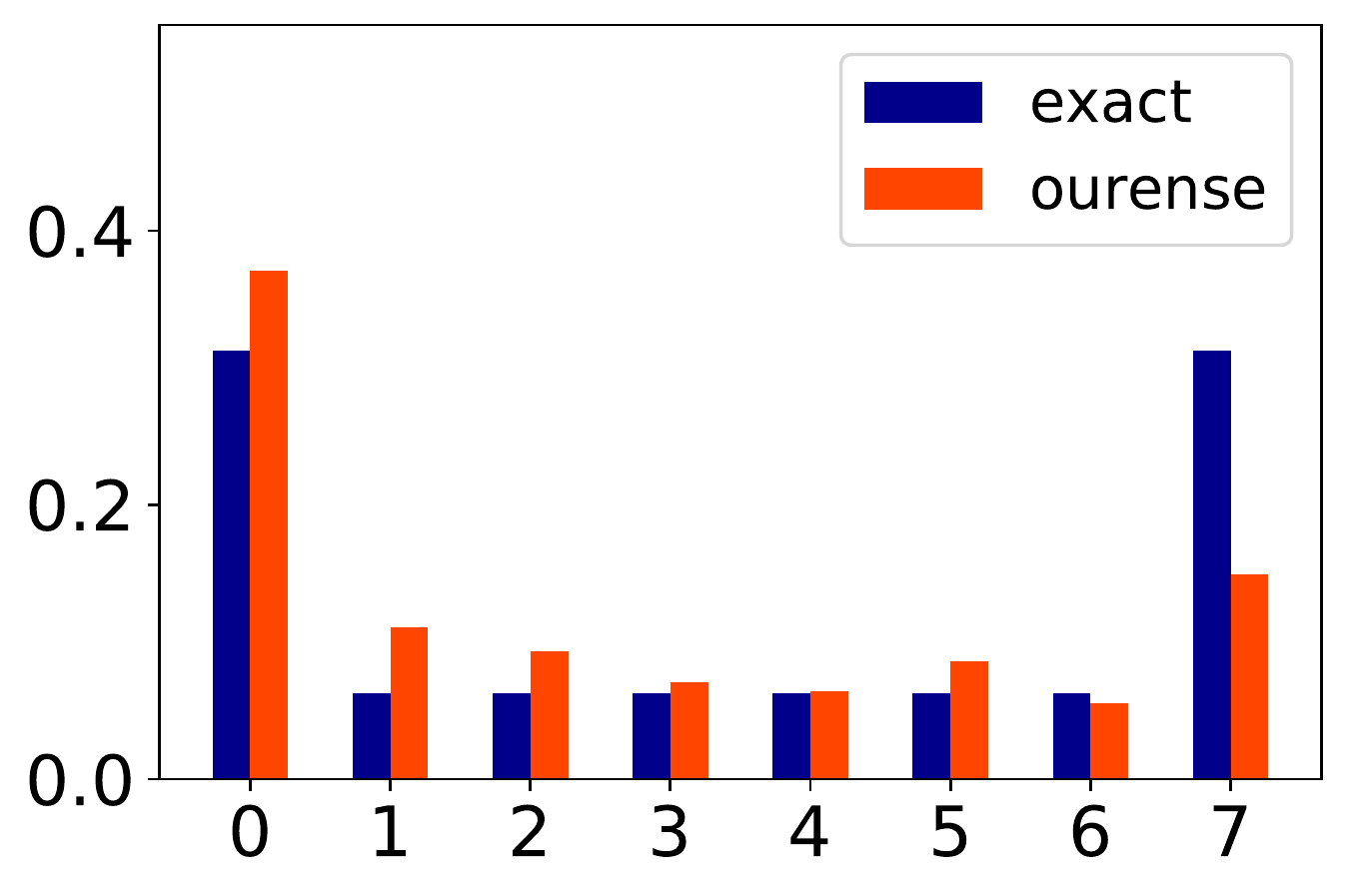}
\includegraphics[scale=0.33]{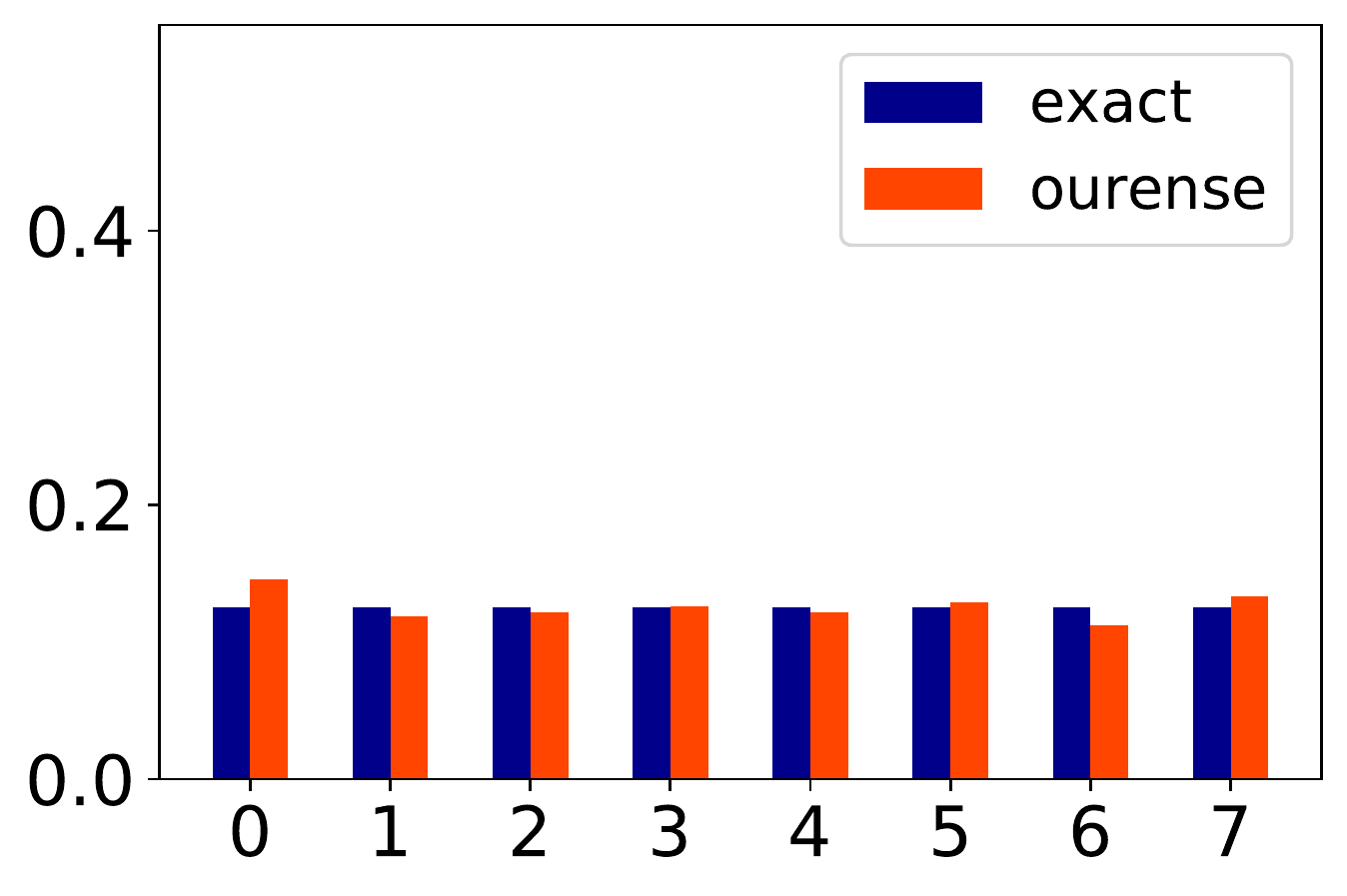}
\includegraphics[scale=0.33]{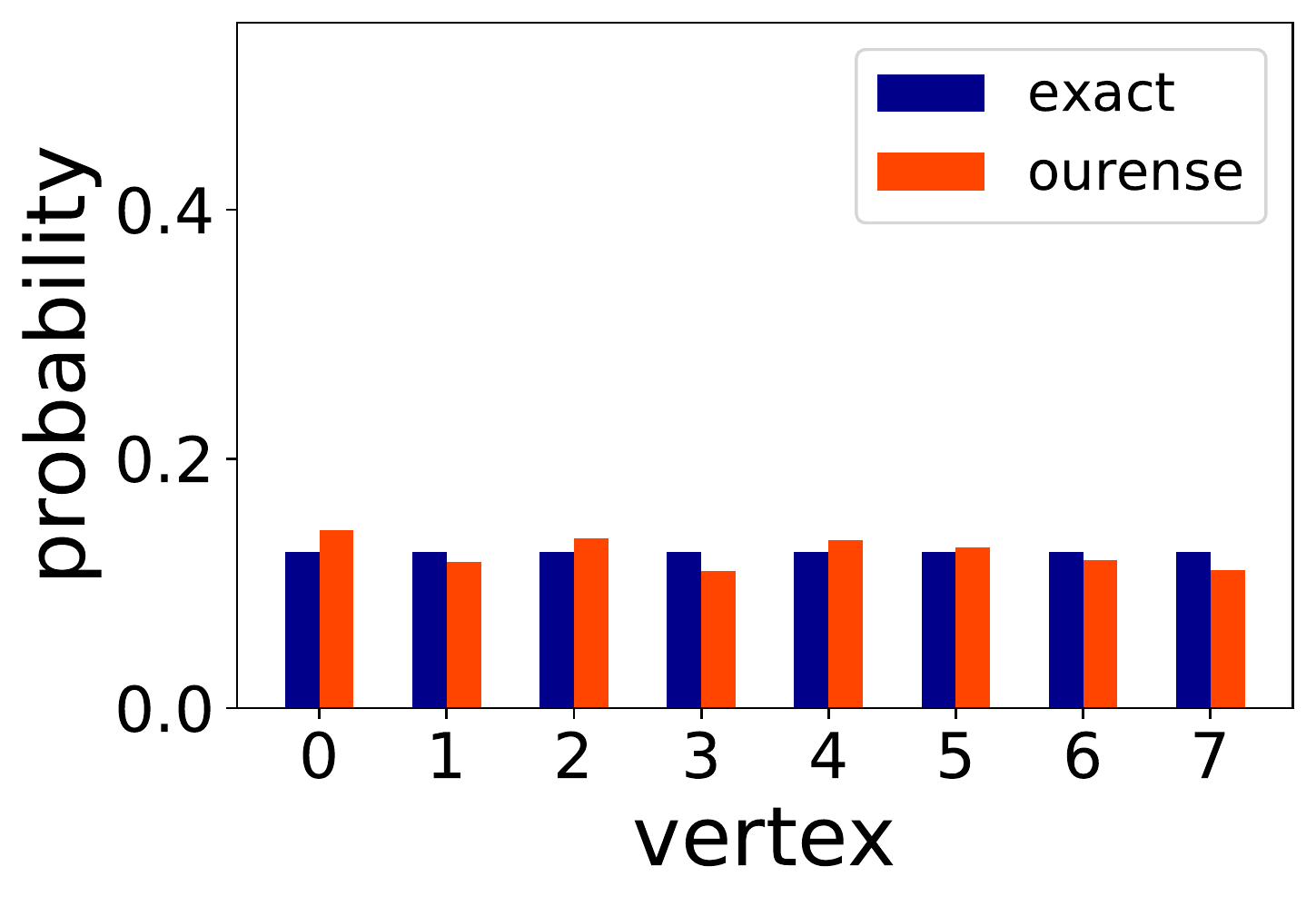}
\includegraphics[scale=0.33]{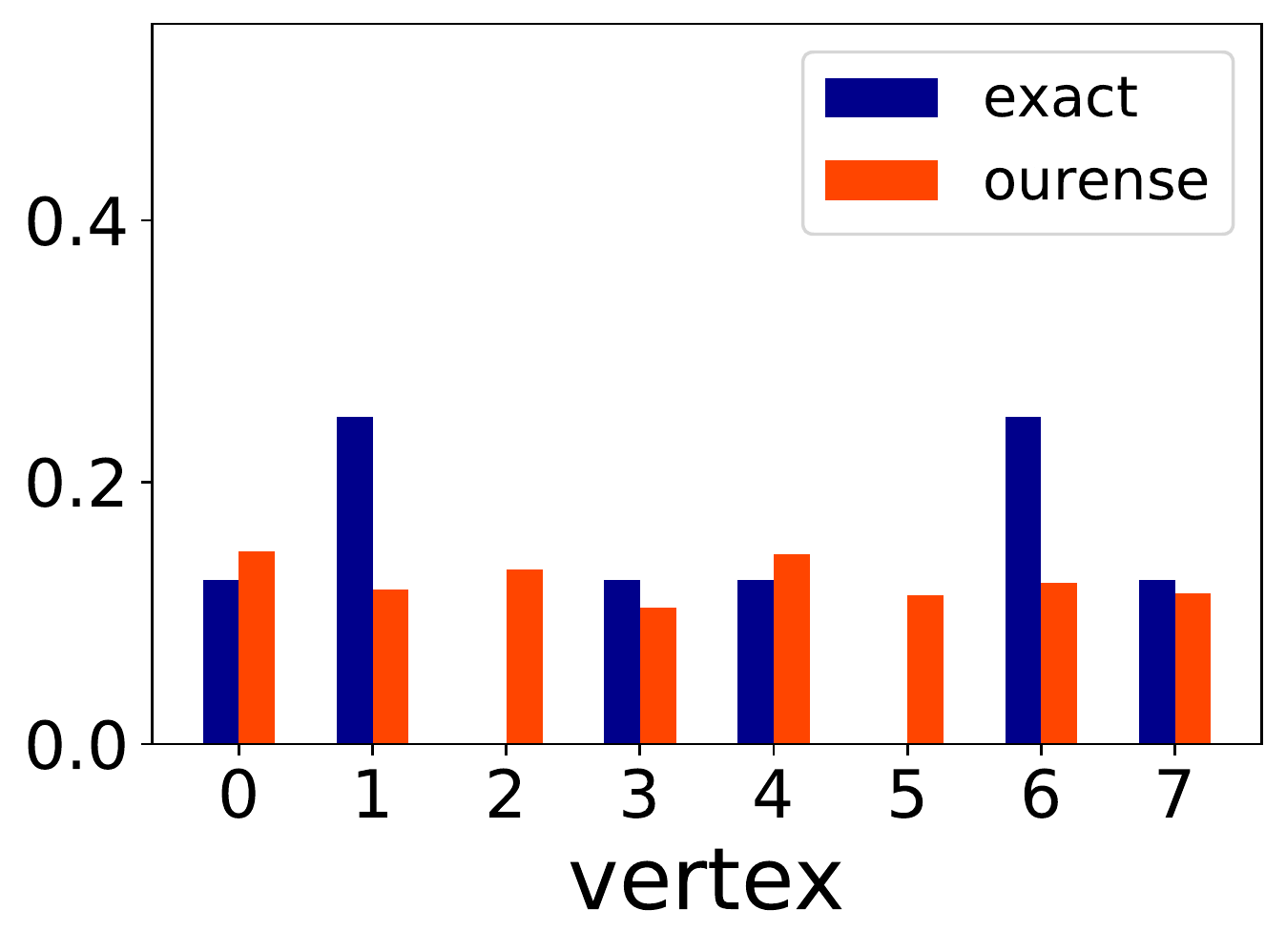}
\includegraphics[scale=0.33]{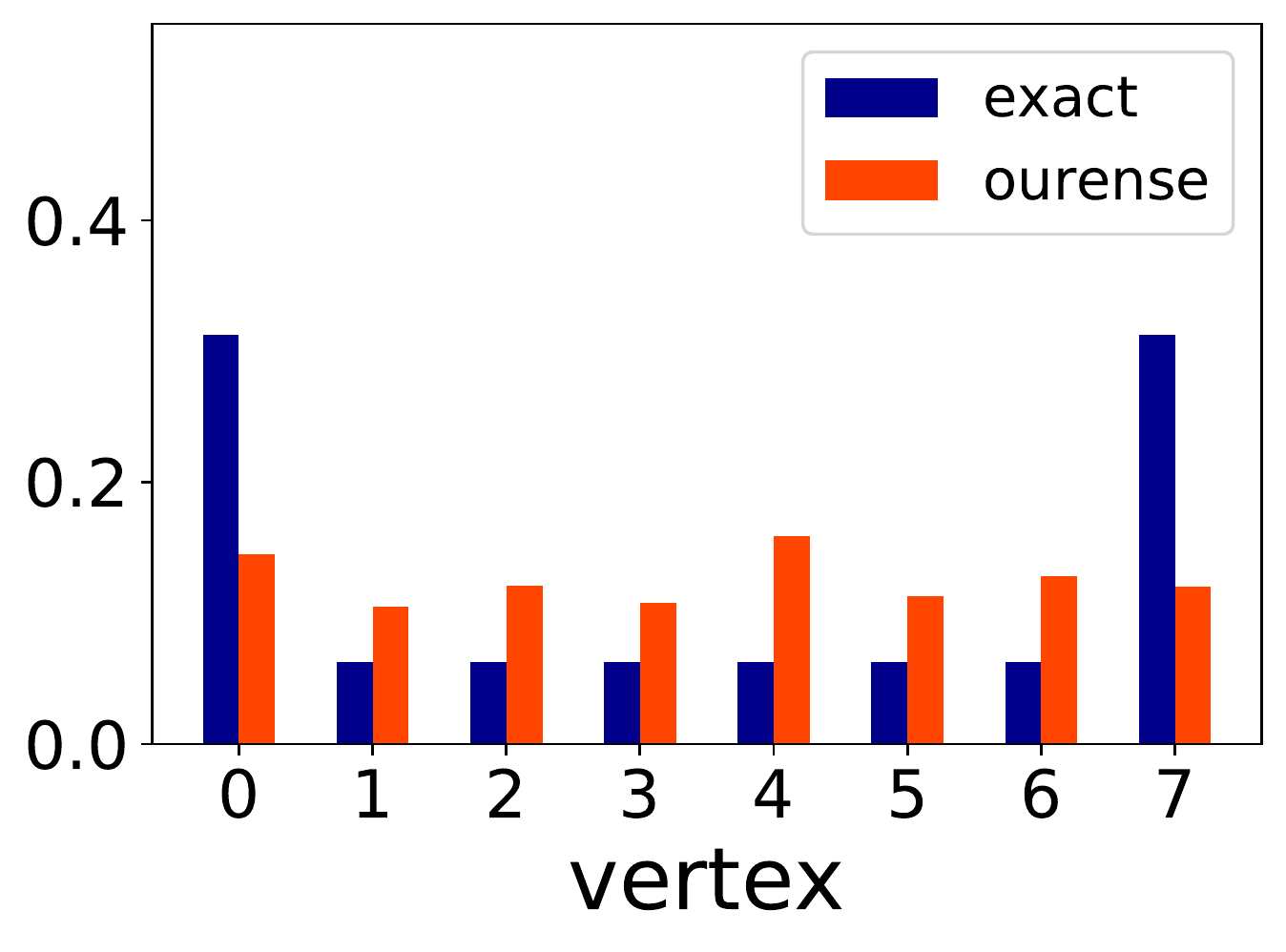}
\caption{Probability distribution of eight steps using exact calculations (blue) and the ourense quantum computer (red) using 3 qubits. The first plot refers to the preparation of the initial state (i.s.) $(\ket{3}+\ket{4})/\sqrt 2$. The following plots are successive steps.}
\label{fig:8-cycle}
\end{figure}

Table~\ref{table:fid3q} shows the corresponding fidelities, where $d$ and $h$ are the total variation and Hellinger distances, respectively. After the sixth step, the fidelity is high but the output of the quantum computer is worthless. This shows that the fidelity is not a good measure when the exact probability distribution is almost flat.

\begin{table}[h!]
\begin{center}
\begin{tabular}{|c|c|c|c|c|c|c|c|c|c|}
\hline
 fidelity & i.s. & step 1 & step 2 & step 3 & step 4 & step 5 & step 6 & step 7 & step 8 \\
\hline
 $1-d$ & 0.927  & 0.891 & 0.864 & 0.896 & 0.823 & 0.965 & 0.956 & 0.710 & 0.639 \\
\hline
 $1-h$ & 0.806  & 0.783 & 0.895 & 0.916 & 0.850 & 0.973 & 0.973 & 0.614 & 0.736 \\
\hline
\end{tabular}
\end{center}
\caption{Fidelities between the probability distributions generated by the ourense quantum computer and the exact simulation for one walker on a 8-vertex cycle up to step 8.}\label{table:fid3q}
\end{table}


\end{document}